\documentclass[journal]{IEEEtran}

\ifCLASSINFOpdf
   \usepackage[pdftex]{graphicx}
\else
   \usepackage[dvips]{graphicx}
\fi
\usepackage[cmex10]{amsmath}
\usepackage{array}
\usepackage{mdwmath}
\usepackage{mdwtab}
\usepackage{eqparbox}
\usepackage[tight,footnotesize]{subfigure}

\usepackage{epstopdf}
\usepackage{mathtools}
\usepackage{color}
\usepackage{slashbox,multirow}
\usepackage{bbding} 
\usepackage{amsfonts}
\usepackage{algorithmic}
\usepackage{cases}

\newtheorem{theorem}{Theorem}
\newtheorem{remark}{Remark}
\newtheorem{lemma}{Lemma}

\begin{document}

\title{Defining Spatial Security Outage Probability for Exposure Region Based Beamforming}

\author{
\IEEEauthorblockN{
Yuanrui Zhang\IEEEauthorrefmark{1},
Youngwook Ko\IEEEauthorrefmark{1},
Roger Woods\IEEEauthorrefmark{1}, 
Alan Marshall\IEEEauthorrefmark{4}
}

\IEEEauthorblockA{
\IEEEauthorrefmark{1}
ECIT, Queen's University Belfast\\
Belfast, UK\\ 
Email: \{yzhang31, r.woods, y.ko\}@qub.ac.uk
}

\IEEEauthorblockA{
\IEEEauthorrefmark{4}
Electrical Engineering and Electronics, University of Liverpool\\
Liverpool, UK\\
Email: Alan.Marshall@liverpool.ac.uk
}
}


\maketitle
\begin{abstract}

With increasing number of antennae in base stations, there is considerable interest in using beamfomining to improve physical layer security, by creating an `exposure region' that enhances the received signal quality for a legitimate user and reduces the possibility of leaking information to a randomly located passive eavesdropper. The paper formalises this concept by proposing a novel definition for the security level of such a legitimate transmission, called the `Spatial Secrecy Outage Probability' (SSOP). By performing a theoretical and numerical analysis, it is shown how the antenna array parameters can affect the SSOP and its analytic upper bound. Whilst this approach may be applied to any array type and any fading channel model, it is shown here how the security performance of a uniform linear array varies in a Rician fading channel by examining the analytic SSOP upper bound.

\end{abstract}

\begin{IEEEkeywords}
Physical layer security, beamforming, exposure region, spatial secrecy outage probability, uniform linear array.
\end{IEEEkeywords}

\IEEEpeerreviewmaketitle

\section{Introduction}\label{sec:intro}


With the proliferation of wireless communications, there is a strong need to provide improved level of security at the physical layer to complement conventional encryption techniques in the higher layers.
Since Wyner established the wiretap channel model and showed the possibility of approaching Shannon's perfect secrecy without a secret key\,\cite{wyner1975wire}, this has been since extended to various channels, such as non-degraded discrete memoryless broadcast channels\,\cite{csiszar1978broadcast}, Gaussian wiretap channels\,\cite{leung1978gaussian}, fading channels\,\cite{barros2006secrecy,bloch2008wireless} and multiple antenna channels\,\cite{shafiee2007achievable,khisti2010secure,khisti2010secure2}.

Wyner's wiretap channel model requires that the legitimate user should have a better channel than the adversarial user, even only for a fraction of realizations in fading channels\,\cite{barros2006secrecy}.
Different users' locations can provide distinction between their channels due to the large-scale path loss relying on user's distance to the transmitter.
However, the role of location in information-theoretic security research has been largely ignored, presumably as users are often assumed to be randomly distributed. 
With the aid of the stochastic geometry theory, the distribution of the random users' locations can be modeled via Poisson point process (PPP),\,\cite{haenggi2009stochastic,chiu2013stochastic} thus encouraging the utilization of location in wireless security. 
For example, ‘ArrayTrack’\,\cite{xiong2013arraytrack} shows how improving granularity can be used to  enhance security\,\cite{xiong2013securearray}. 

This paper mainly investigates the security threat posed by a particular adversarial behavior, i.e., passive eavesdropping, with the classical model where the transmitter (Alice) wishes to transmit to the legitimate user (Bob) in presence of PPP distributed eavesdroppers (Eves).
Alice is equipped with antenna array and performs \textit{beamforming} to enlarge the difference between Bob's and Eve's channels.
Beamforming has been shown to achieve the secrecy capacity in multiple-input-single-output (MISO) channels\,\cite{shafiee2007achievable,khisti2010secure} and has provoked a lot of research\,\cite{mukherjee2010principles,shiu2011physical}.
Essentially, it is a spatial filter that focuses energy in a certain direction or suppresses energy in other directions\,\cite{van1988beamforming}, thereby allowing distinguishing between locations that are either secure or insecure, for the transmission to Bob.
This is important as many applications require security inside an enclosed area, such as different zones in an exhibition hall or different assembly lines in a factory.

In our previous work\,\cite{mypaper,mypaper2}, beamforming is used to create an `exposure region' (ER) to protect the transmission to the legitimate user. 
However, the ER in\,\cite{mypaper,mypaper2} is not based on information-theoretic parameters and lacks the theoretical analysis.
Alternatively, in this paper, the ER is defined by the physical region where any PPP distributed Eve causes secrecy outage to the legitimate transmission in a general channel model, i.e., Rician fading channel.
Then, the spatial secrecy outage probability (SSOP) is defined for the ER based beamforming, which measures the security level of the legitimate transmission based on the ER; this enables an investigation of the role of the array parameters, e.g. number of elements and the direction of emission (DoE) angle, on  physical layer security.

Related work has attempted to create different sorts of physical regions to combat the randomness of both Eve's location and of the fading channel, e.g.,\,\cite{4595864,li2013security,wang2015jamming}.
Whilst the term `exposure region' was coined in\,\cite{4595864}, it referred to received signal quality instead of secrecy outage and lacked information-theoretic analysis whereas in other work,\,\cite{li2013security,wang2015jamming}, the antenna array is overlooked in the information-theoretic analysis.
Since beamforming is performed via antenna arrays, the ER created using beamforming is highly related to the array parameters and can be controlled by changing the array parameters which in turn, affects the SSOP.

The main contributions of this paper are
\begin{itemize}
	\item Definition of the new term called SSOP which is based on the ER where randomly located Eves cause secrecy outage and which measures the security performance in fading channel from the spatial perspective and links with array parameters; it can be applied to existing research to provide information-theoretic analysis and enhanced security performance by taking array parameters into consideration; 
	\item A closed-form expression of the upper bound for the SSOP is obtained to facilitate the theoretical analysis of the security performance, applicable to any array type and fading channel model;
	\item Based on the SSOP, the first investigation of the security performance of ER based beamforming with the uniform linear array (ULA) in a  Rician fading channel with respect to the array parameters is presented. Numerical results reveals that in general, the SSOP increases dramatically as Bob's angle increases; when the number of elements in the array increases, the SSOP converges to a certain value depending on Bob's angle. 
	As for the upper bound, the numerical results show that it is tighter for a smaller number of elements.
\end{itemize}

The paper is organized as follows.
The related work to physical layer security from the physical region perspective is surveyed in Section\,\ref{sec:rw}. In Section\,\ref{sec:sys}, the system model and channel models are demonstrated whereas in Section\,\ref{sec:ERSSOP}, the ER is established, based on which the SSOP and its analytic upper bound are derived.
The SSOP for the ULA and for the Rician channel are analyzed in Sections\,\ref{sec:analysis} and,\ref{sec:analysis2} respectively, along with the tightness of the upper bound.
In Section\,\ref{sec:concl}, the conclusions are given.

\section{Related Work}\label{sec:rw}

Whenever Alice has knowledge of Bob's CSI, beamforming can be used to enhance the received signal quality around Bob and reduce the possibility of leaking information to Eve. As Eve's CSI is generally unknown to Alice, this requires the creation of a physical region either based on the traditional performance metrics, e.g., received power or signal-to-interference-plus-noise ratio\,\cite{4595864,1400008,5357443,sheth2009geo,anand2012strobe,6502515}, or information-theoretic parameters, such as secrecy outage probability (SOP)\,\cite{li2013security,wang2015jamming,sarma2013joint,li2012secure, zheng2014transmission}.

In\,\cite{1400008}, multiple arrays have been used to jointly create a region smaller than that of a single array by dividing the transmitted message and sending it out via multiple arrays in a time-division manner, so that only the user within the jointly created region can receive the complete message.
This idea was extended in\,\cite{4595864,5357443} by encrypting the transmitted message so that only the user within in the jointly created region could decrypt it, with interference sent on some arrays to reduce the effective coverage region. Multiple APs were used in\,\cite{sheth2009geo} to jointly perform beamforming with adaptive transmit power to reduce the joint physical region.

Whilst multiple arrays provide smaller regions, synchronization of the arrays and modifications to higher layer protocol are problematic\,\cite{1400008}.
In \cite{anand2012strobe}, the authors avoid this by using a single array to create a cross-layer design called a STROBE that inserts orthogonal interference which is transmitted simultaneously with the intended data stream, so that Eve cannot decode correctly while Bob remains unaffected by the interference. The work in\,\cite{6502515} designed a specific type of smart array that has two synthesized radiation patterns that can alternatively transmit in a time-division manner overlapping in Bob's direction to provide a full signal transmission whilst reducing signal quality to Eve.

The work based on the traditional performance metrics lacks an information-theoretic analysis, although in\,\cite{4595864,5357443}, the authors define the ER as a performance metric but not using information-theoretic parameters. Work on insecure and secure regions using the information-theoretic parameters has been undertaken on the compromised secrecy region (CSR)\,\cite{li2013security}, secrecy outage region (SOR)\,\cite{wang2015jamming} and vulnerability region (VR)\,\cite{sarma2013joint}, but defined by the region where a certain security goal is not achieved.
On the other hand, the secure regions in\,\cite{li2012secure, zheng2014transmission} are defined by certain security goal being guaranteed.
Despite of the difference in the definition of the physical regions, beamforming and/or artificial noise (AN) are used in the work that is based on information-theoretical parameters, either in the form of antenna arrays\,\cite{wang2015jamming,li2012secure,zheng2014transmission} or in the form of distributed antennas\,\cite{li2013security,sarma2013joint}.

Most reviewed work provides numerical approximations but not the closed-form formulation for these physical regions except\,\cite{wang2015jamming,zheng2014transmission}.
The closed-form formulation of the physical region or its upper bound in this paper can provide analysis with respect to the related aspects, such as array parameters, which can be potentially used for optimization towards higher level security.
In\,\cite{wang2015jamming}, the Rician fading is averaged and treated as a constant in a very large number of antennas systems.
Rayleigh fading generated from simple expressions is considered in\,\cite{zheng2014transmission}, but it is not practical to obtain Bob's location or CSI without the line-of-sight (LOS) component.
It is worth noticing that almost all the reviewed work does not investigate the role of the array parameters in the physical regions.
In\,\cite{yan2014line,yan2014secrecy}, the authors consider some aspect of the array parameters but do not focus on the analysis of the array parameters.

\section{System and Channel Models}\label{sec:sys}

Consider secure communications in wireless local access network,  where the access point (AP), Alice, communicates to a desired receiver (Bob) in presence of passive eavesdroppers (Eves), as shown in Fig.\,\ref{fig:chp3_demon}. 
Suppose that the AP is equipped with an ULA having $N$ antenna elements with a spacing $\mathit{\Delta d}=\lambda/2$, where $\lambda$ is the wavelength of the carrier signal\,\cite{adaptivearraysystems}.
Bob and Eves are assumed to have a single antenna and are simply referred to as a `general user' or a `user' hereinafter, unless otherwise stated.

\begin{figure}
\centering
\includegraphics[scale=1]{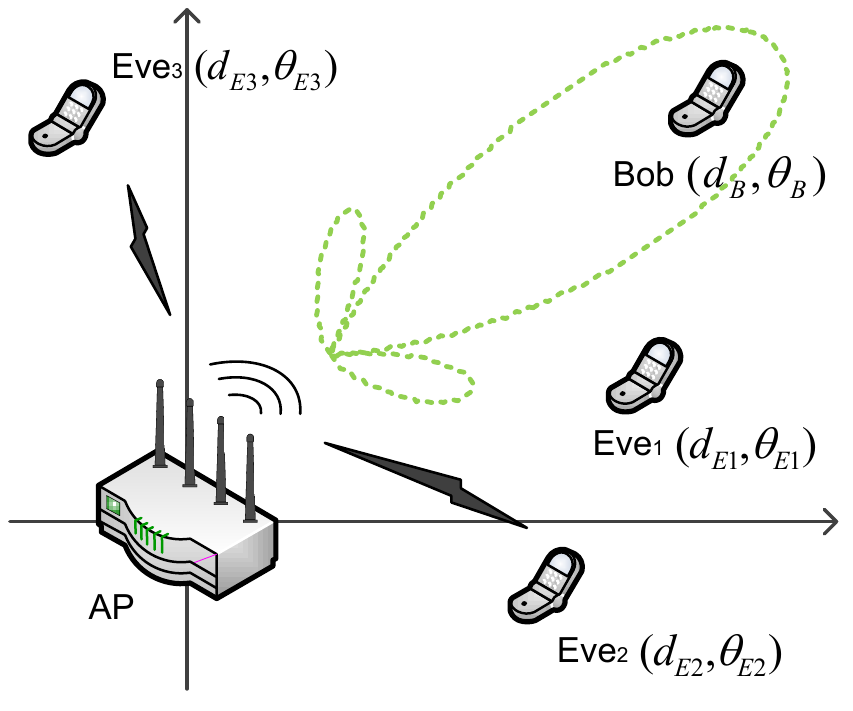}
\caption{An example of a wireless security communications system with one AP, Bob and homogeneous PPP distributed Eves}
\label{fig:chp3_demon}
\end{figure}

We consider that the AP is located at the origin point in polar coordinates, as shown in Fig.\,\ref{fig:chp3_demon}. 
Assume that the users are distributed by a homogeneous PPP, $\Phi_e$, with density $\lambda_e$\,\cite{ghogho2011physical}; the 
user's coordinates are denoted by $z=(d,\theta)$.
Thus, Bob's coordinates are denoted by $z_B=(d_B,\theta_B)$; the $i^{th}$ Eve's coordinate is $z_{Ei}=(d_{Ei},\theta_{Ei}), \forall i\in\mathbb{N}^+$.
The subscripts `$_B$' and `$_E$' are used for Bob and Eves hereinafter.

Given $z_B$, the AP transmits data only towards Bob in the presence of $l$ randomly distributed Eves in every transmit time interval. 
In particular, let $x$ be the modulated symbol with unit power, $\mathbb{E}[|x|^2]=1$, and $P_t$ be its transmit power. 
The transmitted vector, denoted by $\mathbf{u}$, is given by $\mathbf{u}=\sqrt{P_t}\mathbf{w}^*x$, where $\mathbf{w}$ is the beamforming weight vector, i.e.,
$\mathbf{w}=\mathbf{s}(\theta)/\sqrt{N}$, and $\mathbf{s}(\theta)$ is the array steering vector for the ULA,
\begin{align}\label{eq:chp2_steeringvector}
	\mathbf{s}(\theta)=[1,...,e^{-jk\mathit{\Delta d}\sin\theta(i-1)},...,e^{-jk\mathit{\Delta d}\sin\theta(N-1)}]^T,
\end{align}
where $\theta\in[0,2\pi]$ and $k=2\pi/\lambda$.
When $\theta$ is set to $\theta_B$, i.e., $\mathbf{w}=\mathbf{s}(\theta_B)/\sqrt{N}$, the received power at Bob is maximized.
For the 2.4\,GHz Wi-Fi signal, $\lambda=12.5$\,cm.

For a general user at $z=(d,\theta)$, denoted by $\mathbf{h}(z)$, the channel gain vector between the AP and user at $z$ can be decomposed into LOS and non-LOS (NLOS) components, and is expressed by
\begin{align}\label{eq:chp3_CH_Gain}
	\mathbf{h}(z)=d^{-\beta/2}\big(\sqrt{\frac{K}{K+1}}\mathbf{s}(\theta)+\sqrt{\frac{1}{K+1}}\mathbf{g}\big),
\end{align}
where $d^{-\beta/2}$ denotes the large-scale path loss at the distance, $d$, and the path loss exponent $\beta\in[2,6]$; $\mathbf{g}=[g_1,g_2,...,g_N]^T$ represents the NLOS component where every $g_i$ entry is independent and identically distributed (i.i.d.) circularly-symmetric complex Gaussian random variable with zero mean and unit variance, i.e., $g_i\sim\mathcal{CN}(0,1)$; $K$ denotes the factor of the Rician fading.
According to (\ref{eq:chp3_CH_Gain}), the received signal at $z$ can be obtained by
\begin{align}\label{eq:chp3_RX_Signal}
	r(z)=\mathbf{h}^T(z)\mathbf{u}+n_W=\sqrt{\frac{P_t}{d^{\beta}}}\tilde{h}x+n_W,
\end{align}
where $n_W$ is the additive white Gaussian noise with zero mean and variance $\sigma_n^2$ and $\tilde{h}$ is the equivalent channel factor, which is given by
\begin{align}\label{eq:chp3_h_tilde_Ri}
	\tilde{h}&=\big(\sqrt{\frac{K}{K+1}}\mathbf{s}^T(\theta)+\sqrt{\frac{1}{K+1}}\mathbf{g}^T\big)\frac{\mathbf{s}^*(\theta_B)}{\sqrt{N}} \nonumber \\ 
			     &=\sqrt{\frac{K}{K+1}}G(\theta,\theta_B)+\sqrt{\frac{1}{K+1}}\frac{\mathbf{g}^T\mathbf{s}^*(\theta_B)}{\sqrt{N}},
\end{align}
where $G(\theta,\theta_B)$ is the array factor and is given by
\begin{align}\label{eq:chp3_AF_ULA}
G(\theta,\theta_B)&=\frac{1}{\sqrt{N}}\sum_{i=1}^N e^{jk\mathit{\Delta d}(\sin\theta_B-\sin\theta)(i-1)} \nonumber \\
&=\frac{1}{\sqrt{N}}\frac{1-e^{jNk\mathit{\Delta d}(\sin\theta_B-\sin\theta)}}{1-e^{jk\mathit{\Delta d}(\sin\theta_B-\sin\theta)}}.
\end{align}

\begin{remark}\label{prop:chp3_theta_B_range}
The array patterns for $G(\theta,\theta_B)$ at $\pm(\theta_B\pm\pi)$ are symmetric to each other.
Due to this symmetry property of the ULA,, it suffices to study $G(\theta,\theta_B)$ only in $\theta_B\in[0,\frac{\pi}{2}]$.
\end{remark}

Denoted by $\gamma(z)$, the received SNR at $z$, can be found from (\ref{eq:chp3_RX_Signal}), 
\begin{align}\label{eq:chp3_SNR}
	\gamma(z)=\frac{P_t}{\sigma_n^2d^{\beta}}|\tilde{h}|^2.
\end{align}
The channel capacity of the general user at $z$ can be given by
\begin{align}\label{eq:chp3_channelcapacity}
	C(z)=\log_2 [1+\gamma(z)]=\log_2 \Big(1+\frac{P_t}{\sigma_n^2d^{\beta}}|\tilde{h}|^2\Big).
\end{align}
For convenience, let $C_B=C(z_B)$ and $C_{Ei}=C(z_{Ei})$ denote the channel capacities of Bob and the $i$-th Eve hereinafter.
Due to the fact that $|\tilde{h}|^2$ scales with $G(\theta,\theta_B)$, a proper design of $G(\theta,\theta_B)$ can improve $C_B$ while decreasing $C_{Ei}$.

\section{Exposure Region and Spatial Secrecy Outage Probability}\label{sec:ERSSOP}

From (\ref{eq:chp3_channelcapacity}), it can be noticed that $C_{Ei}$ relies on random location $z_{Ei}$ and the small-scale fading $\tilde{h}$.
As a result, one or more Eves could have a higher channel capacity than a certain threshold, leading to the secrecy outage\,\cite{zhou2011rethinking}.
For given Eves' random locations, the exposure region (ER) is mathematically formulated to characterize the above secrecy outage event.
Then the SOP with respect to the ER is evaluated as a measure of the security level.
An upper bound expression for the SSOP is derived to facilitate theoretical analysis.

\subsection{Exposure Region}
\label{chp3:metric:ER}

Let $R_B$ and $R_s$ be the rate of the transmitted codewords and the rate of the confidential information, respectively, then for fixed $R_B$ and $R_s$, a reliable transmission to Bob can be guaranteed when $C_B\geq R_B$.
Secrecy outage event occurs when Eve's channel capacity is higher than the difference $R_B-R_s$ and the probability of such an event is the SOP\,\cite{zhou2011rethinking}.

The geometric meaning is lacking in the above definition of SOP in\,\cite{zhou2011rethinking}.
To characterize the secrecy outage event for the PPP distributed Eves, the ER, denoted by $\Theta$, is defined by the geometric region only where Eves cause the secrecy outage event, i.e., $C_{Ei}>R_B-R_s, \exists z_{Ei}=(d,\theta)\in \Theta$. 
Accordingly, $\Theta$ can be represented by
\begin{align}\label{eq:chp3_ERdef}
	\Theta=\{z:\;C(z)>R_B-R_s\}.
\end{align}
The $i^{th}$ Eve will cause secrecy outage, if and only if $z_{Ei}\in\Theta$.
At the same time, $C_B\geq R_B$ needs to be guaranteed. 

Substituting (\ref{eq:chp3_channelcapacity}) into (\ref{eq:chp3_ERdef}) and rearranging $d$ and $\theta$, $\Theta$ can be transformed into
\begin{align}\label{eq:chp3_erdefinition2}
	\Theta=\{z=(d,\theta):\;d<D(\theta)\},
\end{align}
where 
\begin{align}\label{eq:chp3_ERboundary}
	D(\theta)=\Big[{\frac{P_t|\tilde{h}|^2}{\sigma_n^2(2^{R_B-R_S}-1)}}\Big]^{\frac{1}{\beta}},
\end{align}
$D(\theta)$ is a function only of $\theta$ for a given $\theta_B$ and the contour of $\Theta$.

All locations within $D(\theta)$ have $C(z)>R_B-R_s$, giving a clear geometric meaning, as shown in Fig.\,\ref{fig:chp3_ER_illustration}.
It can be shown from (\ref{eq:chp3_ERboundary}) that $D(\theta)$ (i.e., the shape of $\Theta$) is mainly determined by $|\tilde{h}|^2$.
Thus, $\Theta$ is a dynamic region with shifting boundary whenever $|\tilde{h}|^2$ varies.
When the channel is deterministic, $D(\theta)$ is also deterministic.

\begin{figure}
\centering
\includegraphics[width=3.4in]{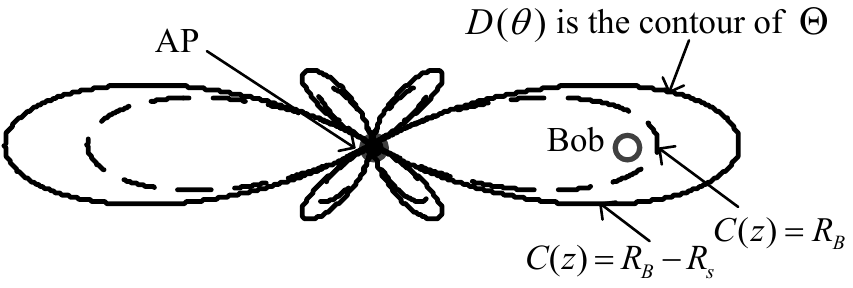}
\caption{Illustration of the ER $\Theta$. $D(\theta)$ is the contour of $\Theta$ for given $\theta_B$, which corresponds to $C(z)=R_B-R_s$; Bob should be within the curve $C(z)=R_B$ to guarantee a reliable transmission.}
\label{fig:chp3_ER_illustration}
\end{figure}

Denoted by $A$, the quantity of $\Theta$ can be measured by the inner area of $D(\theta)$.
Using (\ref{eq:chp3_ERboundary}), $A$ in polar coordinates can be expressed by,
\begin{align}\label{eq:chp3_A1}
A&=\frac{1}{2}\int_0^{2\pi}D^2(\theta)\,\mathrm{d}\theta \nonumber \\
 &=\frac{1}{2}\Big[{\frac{P_t}{\sigma_n^2(2^{R_B-R_S}-1)}}\Big]^{\frac{2}{\beta}}
\int_0^{2\pi}(|\tilde{h}|^2)^{\frac{2}{\beta}}\,\mathrm{d}\theta.
\end{align}
$A$ is measured in\,m$^2$ and depends on $|\tilde{h}|^2$ which can be a function of $G(\theta,\theta_B)$ in the following.
\begin{lemma}\label{le:chp3_h_tilde_squre}
$|\tilde{h}|^2$ can be decomposed by
\begin{align}\label{eq:chp3_h_tilde_square}
	|\tilde{h}|^2 =\frac{KG^2(\theta,\theta_B)}{K+1}+\frac{g_{Re}^2+g_{Im}^2}{K+1}+\frac{2\sqrt{K}G(\theta,\theta_B)}{K+1}g_{Re},
\end{align}
where $g_{Re}$ and $g_{Im}$ are the real and imaginary part of a complex Gaussian random variable $g\sim\mathcal{CN}(0,1)$.
So, $g_{Re}$ and $g_{Im}$ are joint normal distributed variables, i.e., $g_{Re}, g_{Im}\sim \mathcal{N}(0,\frac{1}{2})$.
\end{lemma}
\begin{proof}
In (\ref{eq:chp3_h_tilde_Ri}), let $g$ be the following substitution.
\begin{align}
	g=\frac{\mathbf{g}^T\mathbf{s}^*(\theta_B)}{\sqrt{N}}
\end{align}
where $\mathbf{s}^H(\theta_B)$ is deterministic and each element of $\mathbf{g}$ is an i.i.d. complex Gaussian random variable with zero mean and unit variance.
Therefore, $g$ is a complex Gaussian variable, $g\sim\mathcal{CN}(0,1)$.

Let $g_{Re}$ and $g_{Im}$ denote the real and imaginary part of $g$, where $g_{Re}$ and $g_{Im}$ are joint normal variables, i.e., $g_{Re}, g_{Im}\sim \mathcal{N}(0,\frac{1}{2})$.
Thus, 
\begin{align}
	\tilde{h}=\sqrt{\frac{K}{K+1}}G(\theta,\theta_B)+\sqrt{\frac{1}{K+1}}g_{Re}+j\sqrt{\frac{1}{K+1}}g_{Im}.
\end{align}
Then, $|\tilde{h}|^2$ can be obtained by
\begin{align}
	|\tilde{h}|^2 &= \Big[\sqrt{\frac{K}{K+1}}G(\theta,\theta_B)+\sqrt{\frac{1}{K+1}}g_{Re}\Big]^2+ \frac{1}{K+1}g_{Im}^2 \nonumber \\
	&=\frac{KG^2(\theta,\theta_B)}{K+1}+\frac{g_{Re}^2+g_{Im}^2}{K+1}+\frac{2\sqrt{K}G(\theta,\theta_B)g_{Re}}{K+1}.
\end{align}
\end{proof}

A reliable transmission is guaranteed for Bob, if Bob is inside the dashed curve in Fig.\,\ref{fig:chp3_ER_illustration}, i.e., $C_B>R_B$.
A secrecy outage event only occurs when $z_{Ei}\in\Theta$. 
Intuitively, given that Bob's reliable transmission is guaranteed, the smaller $A$ is, the smaller number of Eves are statistically located in $\Theta$, leading to less occurrence of the secrecy outage.

\subsection{Spatial Secrecy Outage Probability}
\label{chp3:metric:ssop}

Any Eve at $z_{Ei}\in\Theta$ causes $C_{Ei}>R_B-R_s$ and this is referred to as a spatial secrecy outage (SSO) event with respect to the ER.
The \textit{spatial secrecy outage probability} (SSOP) can be defined by the probability that any Eve is located inside $\Theta$.
To the best of our knowledge, the SSOP provides distinctive measure of the ER based security over the conventional SOP which does not have dynamic geometric implication; the SSOP emphasizes the secrecy outage caused by the spatially distributed Eves within a dynamic $\Theta$.

We quantify the SSOP, denoted by $p$, to measure the secrecy performance.
Particularly for given PPP-distributed Eves, the probability that $m$ Eves are located inside $\Theta$ (with its area quantity $A$) is given by
\begin{align}\label{eq:chp3_PPP}
	\text{Prob}\{m\;\text{Eves in}\;\Theta\}=\frac{(\lambda_eA)^m}{m!}e^{-\lambda_eA}.
\end{align}
Using (\ref{eq:chp3_A1}) and (\ref{eq:chp3_PPP}), $p$ can be quantitatively measured by referring to `no secrecy outage' event that no Eves are located inside $\Theta$ and is given by
\begin{align}\label{eq:chp3_SSOP}
	p=1-\text{Prob}\{0\;\text{Eve in}\;\Theta\}=1-e^{-\lambda_eA}.
\end{align}
It can be seen from (\ref{eq:chp3_SSOP}) that the smaller $p$ is, the less the spatial secrecy outage occurs.
This results in the more secure transmission to Bob.
For a given $\lambda_e$, $p$ decreases along with $A$.

\newcounter{MYtempeqncnt}
\begin{figure*}[!t]
\normalsize
\setcounter{MYtempeqncnt}{\value{equation}}
\setcounter{equation}{18}
\begin{numcases}{\bar{p}=}
1-\int_{-\infty}^{\infty}\int_{-\infty}^{\infty} \text{exp}\Big\{-\frac{\lambda_e}{2}c_0^{\frac{2}{\beta}}\int_0^{2\pi}\Big[\frac{KG^2(\theta,\theta_B)}{K+1} &\nonumber \\
 +\frac{x^2+y^2}{K+1}+\frac{2\sqrt{K}G(\theta,\theta_B)}{K+1}x\Big]^{\frac{2}{\beta}}\,\mathrm{d}\theta\Big\} \frac{e^{-(x^2+y^2)}}{\pi} \,\mathrm{d}x\,\mathrm{d}y, & $K\in(0,\infty)$   \label{eq:chp3_meanSSOP_Ri_2} \\
1-\text{exp}\Big\{-\frac{\lambda_e}{2}c_0^{\frac{2}{\beta}}\int_0^{2\pi}[G^2(\theta,\theta_B)]^{\frac{2}{\beta}}\,\mathrm{d}\theta\Big\}, & $K=\infty$ \label{eq:chp3_SSOP_De} \\
1-\int_{-\infty}^{\infty}\int_{-\infty}^{\infty} \text{exp}\Big\{-\lambda_e\pi c_0^{\frac{2}{\beta}}(x^2+y^2)^{\frac{2}{\beta}}\Big\} \frac{e^{-(x^2+y^2)}}{\pi} \,\mathrm{d}x\,\mathrm{d}y, & $K=0$,\label{eq:chp3_meanSSOP_Ra}
\end{numcases}
\addtocounter{equation}{1}
\begin{align}\label{eq:chp3_meanSSOP_Ri_temp}
	\bar{p}=\mathbb{E}_{g_{Re},g_{Im}}[p]=1-\int_{-\infty}^{\infty}\int_{-\infty}^{\infty} \text{exp}\Big\{-\frac{\lambda_e}{2}c_0^{\frac{2}{\beta}}\int_0^{2\pi}\Big[\frac{KG^2(\theta,\theta_B)+x^2+y^2+2\sqrt{K}G(\theta,\theta_B)x}{K+1}\Big]^{\frac{2}{\beta}}\,\mathrm{d}\theta\Big\} f_{g_{Re}}(x)f_{g_{Im}}(y) \,\mathrm{d}x\,\mathrm{d}y.
\end{align}
\setcounter{equation}{\value{MYtempeqncnt}}
\hrulefill
\vspace*{4pt}
\end{figure*}

Notice that $p$ in (\ref{eq:chp3_SSOP}) depends on the equivalent channel factor $\tilde{h}$ via $A$.
Due to the fact that $\tilde{h}$ is random channel fading, it is more interesting to study the expectation of $p$, which reflects the averaged SSOP, which is denoted by $\bar{p}$ and can be calculated by
\addtocounter{equation}{0}
\begin{align}\label{eq:chp3_meanSSOP_Ri_0}
	\bar{p}&=\mathbb{E}_{|\tilde{h}|}[p]=1-\mathbb{E}_{|\tilde{h}|}[e^{-\lambda_eA}].
\end{align}

\begin{theorem}\label{th:chp3_ssop}
Given $A$ in (\ref{eq:chp3_A1}), $\bar{p}$ in (\ref{eq:chp3_meanSSOP_Ri_0}) can be expressed by (\ref{eq:chp3_meanSSOP_Ri_2}) to (\ref{eq:chp3_meanSSOP_Ra}) at the top of the page, 
where $\lambda_e$ is the density of Eves, $c_0=\frac{P_t}{\sigma_n^2(2^{R_B-R_S}-1)}$ is deterministic, $\beta$ is the path loss exponent, $K$ is the Rician factor, $G(\theta,\theta_B)$ is the array factor when the DoE angle is Bob's angle $\theta_B$.
\end{theorem}

\begin{proof}
First, substituting $c_0$ into (\ref{eq:chp3_A1}), $A$ can be simplified into
\addtocounter{equation}{3}
\begin{align}\label{eq:chp3_A2}
A=\frac{1}{2}\int_0^{2\pi}(c_0|\tilde{h}|^2)^{\frac{2}{\beta}}\,\mathrm{d}\theta.
\end{align}

Substituting (\ref{eq:chp3_h_tilde_square}) and (\ref{eq:chp3_A2}) into (\ref{eq:chp3_meanSSOP_Ri_0}), $\bar{p}$ can be calculated by (\ref{eq:chp3_meanSSOP_Ri_temp}) at the top of the page.
For normal distribution,
\addtocounter{equation}{1}
\begin{align}
	 f_{g_{Re}}(x)&=\frac{1}{\sqrt{\pi}}e^{-x^2} \label{eq:chp3_normal_pdf_x},\\
	 f_{g_{Im}}(y)&=\frac{1}{\sqrt{\pi}}e^{-y^2} \label{eq:chp3_normal_pdf_y}.
\end{align}
(\ref{eq:chp3_meanSSOP_Ri_2}) can be obtained by substituting (\ref{eq:chp3_normal_pdf_x}) and (\ref{eq:chp3_normal_pdf_y}) into (\ref{eq:chp3_meanSSOP_Ri_temp}).
Take the limit of $K\to\infty$ and $K\to 0$, (\ref{eq:chp3_SSOP_De}) and (\ref{eq:chp3_meanSSOP_Ra}) can be obtained, respectively.
Thus, the proof is completed.
\end{proof}

It is worth pointing out that for the deterministic channel ($K\to\infty$), $\bar{p}$ in (\ref{eq:chp3_SSOP_De}) is mainly decided by $G(\theta,\theta_B)$, while for the Rayleigh channel ($K=0$),  $\bar{p}$ in (\ref{eq:chp3_meanSSOP_Ra}) is shown not to contain $G(\theta,\theta_B)$, as there is no LOS component in Rayleigh fading channel.
$\bar{p}$ in Theorem\,\ref{th:chp3_ssop} is complex and can be numerically calculated.
However, it is not tractable to obtain in closed-form expression, except for the deterministic channel when $\beta=2$.
In the next subsection, upper bound expression for $\bar{p}$ will be derived in closed-form to facilitate detailed theoretical analysis.

\subsection{Upper Bound Expression for Averaged SSOP}
\label{chp3:metric:bounds}

To obtain the analytic upper bound expression, consider two major obstacles.
First, let $X_{\theta}=c_0|\tilde{h}|^2$.
$A$ in (\ref{eq:chp3_A1}) can be written in terms of $X_{\theta}$ as
\begin{align}\label{eq:chp3_A3}
A=\frac{1}{2}\int_0^{2\pi}X_{\theta}^{\frac{2}{\beta}}\,\mathrm{d}\theta.
\end{align}
$X_{\theta}$ relies on the array factor $G(\theta,\theta_B)$.
It is not straightforward to solve the integral when $\beta>2$.
The other obstacle is that $\mathbb{E}_{\tilde{h}}[e^{-\lambda_eA}]$ in (\ref{eq:chp3_meanSSOP_Ri_0}) is not mathematically tractable due to the composite array factor and Rician fading channels.

To overcome the aforementioned obstacles, we aim to obtain the moments of $|\tilde{h}|^2$.
Denoted by $\bar{p}_{up}$, the upper bound for $\bar{p}$ can be obtained via the moments of $|\tilde{h}|^2$ using two instances of Jensen's Inequality.
\begin{align}\label{eq:chp3_JI_1}
	\mathbb{E}[e^X]\geq e^{\mathbb{E}[X]},
\end{align}
where $X$ is a random variable.
The equality holds if and only if $X$ is a deterministic value.
The other one involved is expressed by
\begin{align}\label{eq:chp3_JI_2}
  \mathbb{E}[X^{\frac{2}{\beta}}]\leq (\mathbb{E}[X])^{\frac{2}{\beta}},
\end{align}
where $X$ is a random variable and $\beta\geq 2$.
The equality holds when $\beta=2$ for any $X$.

\begin{theorem}\label{th:sec4_p_up}
For given $\lambda_e$ and $K$, $\bar{p}_{up}$ can be derived using (\ref{eq:chp3_JI_1}) and (\ref{eq:chp3_JI_2}) and is expressed by
\begin{numcases}{\bar{p}_{up}=}
1-\text{exp}\Big\{-\lambda_e\pi c_0^{\frac{2}{\beta}}\Big[\frac{KA_0}{2\pi(K+1)} \nonumber \\
+\frac{1}{(K+1)}\Big]^{\frac{2}{\beta}}\Big\}, \quad\quad\quad\quad\quad\quad\, K\in(0,\infty)   \label{eq:chp3_meanSSOP_up_2} \\
1-\text{exp}\Big[-\lambda_e\pi \Big(\frac{c_0}{2\pi}A_0 \Big)^{\frac{2}{\beta}} \Big],\quad\quad K=\infty \label{eq:chp3_SSOP_De_up} \\
1-\text{exp}(-\lambda_e\pi c_0^{\frac{2}{\beta}}), \quad\quad\quad\quad\quad\quad K=0,\label{eq:chp3_SSOP_Ra_up}
\end{numcases}
where $A_0$ denotes the pattern area and is given by,
\begin{align}
	A_0&=\int_0^{2\pi} G^2(\theta,\theta_B)\,\mathrm{d}\theta \label{eq:chp3_A_0} \\
	   &=2\pi+4\pi\sum_{n=1}^{N-1} \frac{N-n}{N}J_0(k\mathit{\Delta d}n)\cos(k\mathit{\Delta d}n\sin\theta_B), \label{eq:chp3_A_0L}
\end{align}
where $J_0(x)$ is the Bessel function of the first kind with order zero, and $k=2\pi/\lambda$ is a constant.
\end{theorem}

\begin{proof}
see Appendix\,\ref{appdx:bessel:owejg}.
\end{proof}

It is worth mentioning that $A_0$ in (\ref{eq:chp3_A_0}) is a general expression to be applied to any type of array (e.g., linear array, circular array).
For the ULA, we can find approximations for $A_0$ in (\ref{eq:chp3_A_0L}), because $J_0(x)$ has a decreasing envelope with the maximum value $J_0(0)=1$ at $x=0$, and approaches zero when $x$ increases.
This will facilitate the analytical analysis for $\bar{p}_{up}$, which in turn provides guidance for the analysis of $\bar{p}$, especially if $\bar{p}_{up}$ is close to $\bar{p}$.

\begin{remark}\label{prop:chp3_SSOP_up_analysis}
Notice that the inequalities in (\ref{eq:chp3_JI_1}) and (\ref{eq:chp3_JI_2}) are used to derive $\bar{p}_{up}$.
When $K=\infty$ and $\beta=2$, the equality holds for both (\ref{eq:chp3_JI_1}) and (\ref{eq:chp3_JI_2}); thus, $\bar{p}_{up}=\bar{p}$.
This can be verified by substituting $\beta=2$ into (\ref{eq:chp3_SSOP_De}) and (\ref{eq:chp3_SSOP_De_up}).
Similarly, when $K=\infty$, the equality holds only for (\ref{eq:chp3_JI_1}); thus, $\bar{p}_{up}$ is tighter when $\beta=2$ than that when $\beta>2$ according to (\ref{eq:chp3_JI_2}).
When $\beta=2$, the equality holds only for (\ref{eq:chp3_JI_2}); thus, $\bar{p}_{up}$ is tighter for $K=\infty$ than that for $K<\infty$ according to (\ref{eq:chp3_JI_1}).
For other cases, the tightness of $\bar{p}_{up}$ is not straightforward.
The numerical results of $\bar{p}_{up}$ for different $K$ and $\beta$ will be given in Section\,\ref{chp3:result:mnbv}.
\end{remark}

\begin{remark}
Both $\bar{p}$ in (\ref{eq:chp3_meanSSOP_Ri_2})-(\ref{eq:chp3_meanSSOP_Ra}) and $\bar{p}_{up}$ in (\ref{eq:chp3_meanSSOP_up_2})-(\ref{eq:chp3_SSOP_Ra_up}) are positively correlated with the transmit power $P_t$ via $c_0$. 
It is worth noticing that $P_t$ influences the SSOP being independent of the array parameters ($N$ and $\theta_B$).
Therefore, in this paper, when studying the impact of the array parameters, $P_t$ is treated as constant within the constant $c_0$.
\end{remark}

\section{Impact of ULA Parameters on Averaged SSOP}\label{sec:analysis}

In this section, we focus on the impact of ULA parameters (i.e., $N$ and $\theta_B$) on $A_0$ and thus the averaged SSOP $\bar{p}$.
To this end, we consider the asymptotic case when $K\to\infty$ and $N\to\infty$.
As stated in Remark\,\ref{prop:chp3_SSOP_up_analysis}, when $K\to\infty$ and $\beta=2$, we have $\bar{p}_{up}=\bar{p}$.
According to (\ref{eq:chp3_SSOP_De}) and (\ref{eq:chp3_SSOP_De_up}), it gives
\begin{align}\label{eq:chp3_p_De_beta_is_2}
	\bar{p}=\bar{p}_{up}=1-\text{exp}(-\frac{\lambda_ec_0}{2}A_0 ).
\end{align}
As seen in (\ref{eq:chp3_p_De_beta_is_2}), $\bar{p}_{up}$ (i.e., $\bar{p}$) monotonically increases with $A_0$.
Thus, it suffices to analyze the behavior of $A_0$.
Detailed numerical results for $\bar{p}$ and $\bar{p}_{up}$ for generalized values of $K$ and $\beta$ will be shown in Section\,\ref{chp3:result:wier}.


\subsection{Impact of $\theta_B$}

As stated in Remark\,\ref{prop:chp3_theta_B_range}, the range of $\theta_B\in[0,\frac{\pi}{2}]$ is concerned.
First, let $A_{0,n}$, for $n=1,...,N-1$, denote the summation term in (\ref{eq:chp3_A_0L}) and it is given by
\begin{align}\label{eq:chp3_A_0L_n}
	A_{0,n}=4\pi\frac{N-n}{N}J_0(k\mathit{\Delta d}n)\cos(k\mathit{\Delta d}n\sin\theta_B).
\end{align}
When $\mathit{\Delta d}=0.5\lambda$, (\ref{eq:chp3_A_0L_n}) can be written as
\begin{align}\label{eq:chp3_A_0L_n_2}
	A_{0,n}=4\pi\frac{N-n}{N}J_0(n\pi)\cos(n\pi\sin\theta_B).
\end{align}
Using (\ref{eq:chp3_A_0L}) and (\ref{eq:chp3_A_0L_n_2}), $A_0$ can be represented by
\begin{align}\label{eq:chp3_A_0L_2}
	A_0&=2\pi+\sum_{n=1}^{N-1}A_{0,n} \nonumber \\
	   &=2\pi+4\pi\sum_{n=1}^{N-1}\frac{N-n}{N}J_0(n\pi)\cos(n\pi\sin\theta_B)
\end{align}

When $N=8$ and $\mathit{\Delta d}=0.5\lambda$, the envelope of the components in (\ref{eq:chp3_A_0L_2}) is shown in Fig.\,\ref{fig:chp3_besselj_ULA}.
In the upper plot of Fig.\,\ref{fig:chp3_besselj_ULA}, $J_0(n\pi)$ is shown to decrease as $n$. 
The lower plot depicts the decreasing envelope of $A_{0,n}$, i.e., $\frac{N-n}{N}J_0(n\pi)$, with $n$.
When $n=1$, $A_{0,1}$ is the largest;
when $n=7$, $A_{0,7}$ is negligible.

\begin{figure}
\centering
\includegraphics[width=3.4in]{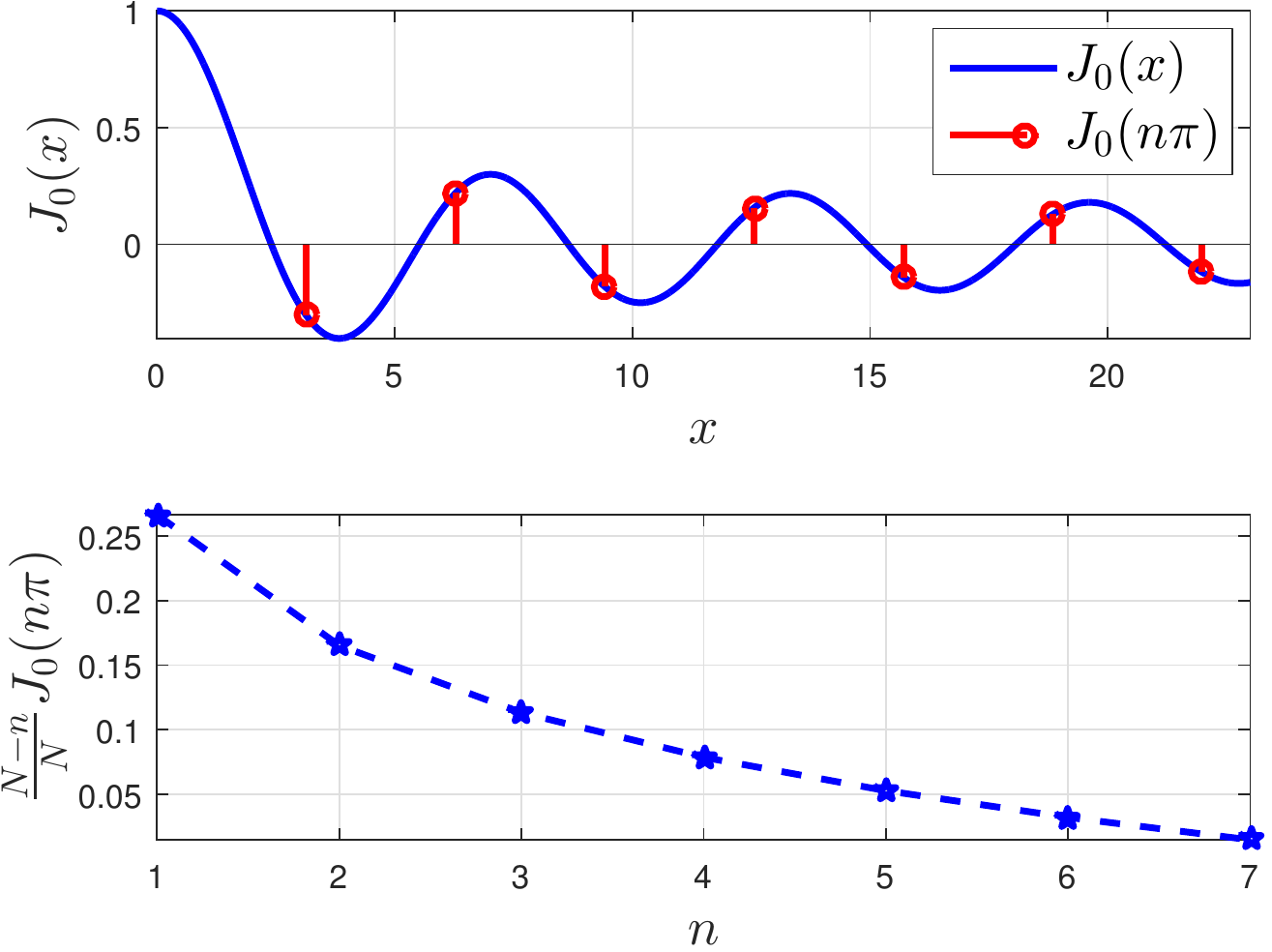}
\caption{Behavior of $J_0(n\pi)$ and $\frac{N-n}{N}J_0(n\pi)$ in $A_{0,n}$ for $n\leq N-1$ when $N=8$ antenna elements are used. As $n$ increases, both terms become less significant in $A_0$ (i.e., $\bar{p}_{up}$).}
\label{fig:chp3_besselj_ULA}
\end{figure}

As a result, for given $N$, we can approximate $A_0$ by considering the first few dominant terms.
Especially in the case when $\mathit{\Delta d}=0.5\lambda$, $A_{0,1}$ is dominant and it suffices to approximate $A_0$ using only $A_{0,1}$, i.e., 
\begin{align}\label{eq:chp3_A_0L_approx}
	A_0&\approx 2\pi+4\pi\frac{N-1}{N}J_0(\pi)\cos(\pi\sin\theta_B)\nonumber \\
	&=2\pi+\mathcal{O}(J_0(\pi)\cos(\pi\sin\theta_B)).
\end{align}
Using (\ref{eq:chp3_p_De_beta_is_2}) and (\ref{eq:chp3_A_0L_approx}), when $K\to\infty$, $\bar{p}_{up}$ can be asymptotically approximated by
\begin{align}\label{eq:chp3_p_De_approx_betais2}
	\lim_{K\to\infty}  \bar{p}_{up}=1-\text{exp}\{-\lambda_e c_0 \pi -\mathcal{O}(J_0(\pi)\cos(\pi\sin\theta_B))\},
\end{align}
where $\mathcal{O}(\cdot)$ denotes the big O notation.

From (\ref{eq:chp3_A_0L_approx}) and (\ref{eq:chp3_p_De_approx_betais2}), it can be seen that for any given $N$, $\bar{p}_{up}$ increases along with $\theta_B$ in the range $\theta_B\in[0,\pi/2]$, because $\cos(\pi\sin\theta_B)$ decreases from $1$ to $-1$ when $\theta_B$ increases from $0$ to $\frac{\pi}{2}$ and $J_0(\pi)<0$ as illustrated by the upper plot in Fig.\,\ref{fig:chp3_besselj_ULA}.

Fig.\,\ref{fig:chp3_p_DoE_De_ULA} depicts the impact of $\theta_B$ on $A_{0,n}$ and $\bar{p}_{up}$. 
In particular, $A_{0,n}$ is the components of $A_0$, which $\bar{p}_{up}$ relies on according to (\ref{eq:chp3_p_De_beta_is_2}).
For the illustrations, we use the ULA with $N=8$ and $\mathit{\Delta d}=0.5\lambda$.
In the left plot, $A_{0,1}$ among $A_{0,n}$, for $n=\{1,2,3\}$, has the largest variation from $\theta_B=0^{\circ}$ to $\theta_B=90^{\circ}$.
The variation of $A_{0,n}$ in $\theta_B\in[0,\pi/2]$ becomes smaller at larger $n$.

\begin{figure}
\centering
\includegraphics[width=3.4in]{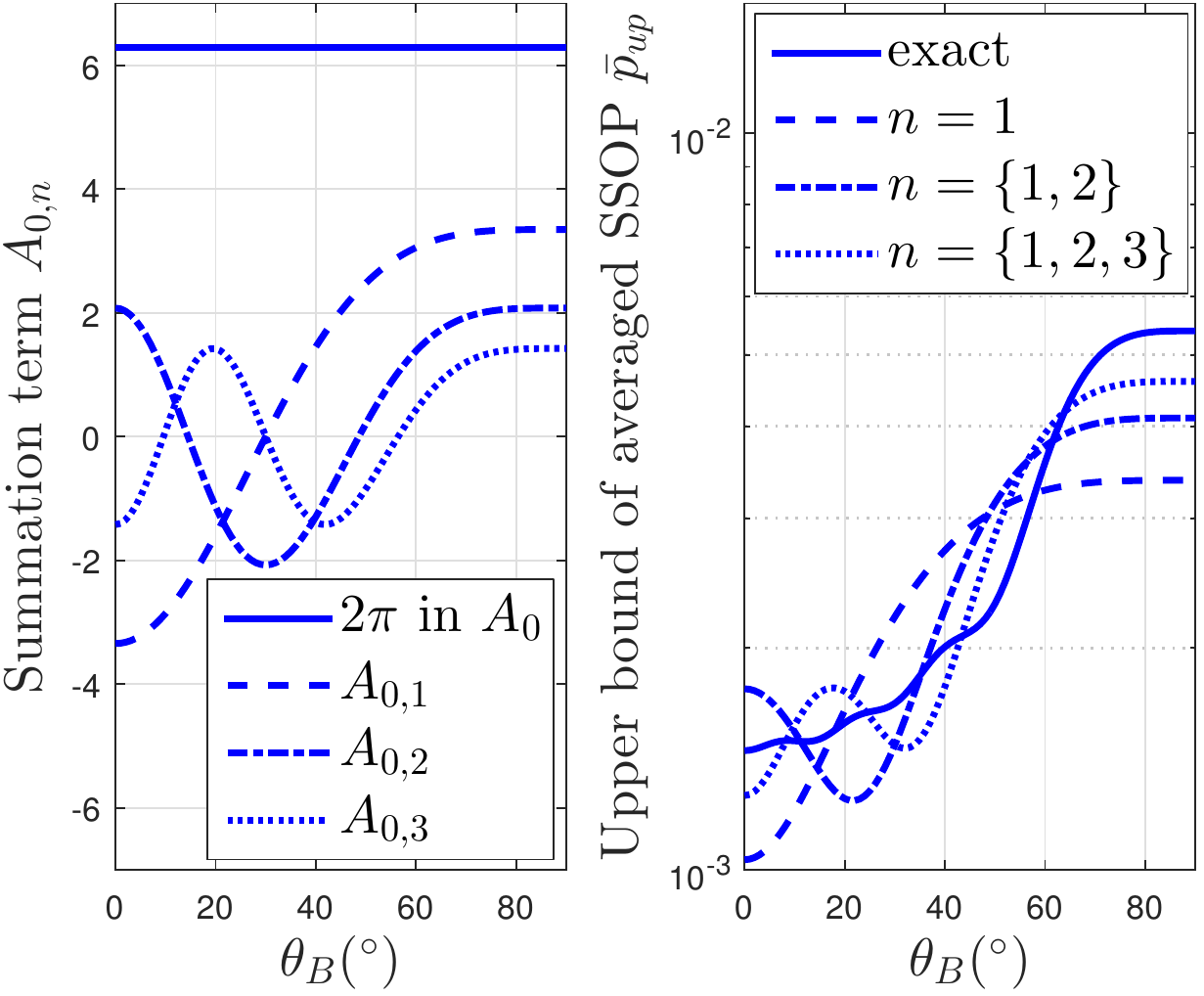}
\caption{Impact of $\theta_B$ on $A_{0,n}$ and $\bar{p}_{up}$. Left plot: $A_{0,n}$ versus Bob's angle $\theta_B$. Right plot: exact value and approximations of $\bar{p}_{up}$ versus $\theta_B$when $N=8$, $\mathit{\Delta d}=0.5\lambda$. $P_t/\sigma_n^2=15$\,dB, $R_B=3.4594$\,bps/Hz, $R_s=1$\,bps/Hz, $\lambda_e=1\times10^{-4}$}
\label{fig:chp3_p_DoE_De_ULA}
\end{figure}

In the right plot of Fig.\,\ref{fig:chp3_p_DoE_De_ULA}, the exact $\bar{p}_{up}$ is shown in comparison to its various approximations:
when $n=1$, the approximated $\bar{p}_{up}$ in (\ref{eq:chp3_p_De_approx_betais2}) is used, which relies on $A_{0,1}$;
when $n=\{1,2\}$, the approximated $\bar{p}_{up}$ in (\ref{eq:chp3_p_De_beta_is_2}) relies on $A_{0,n}$ in (\ref{eq:chp3_A_0L_2}), and so forth.
It can be seen in Fig.\,\ref{fig:chp3_p_DoE_De_ULA} that when $n=1$, the approximation already captures the increasing trend of the exact value.
With more values of $A_{0,n}$, the approximation becomes closer to the exact value.

It is worth noticing from Fig.\,\ref{fig:chp3_p_DoE_De_ULA} that $A_{0,n}$, for $n>2$, is not monotonic in the range $\theta_B\in[0,\frac{\pi}{2}]$.
However, for $n>2$, $A_{0,n}$, is less dominant than $A_{0,1}$. 
Overall, the exact value of $\bar{p}_{up}$ is depicted to have a monotonic increasing relationship with $\theta_B$ in general.

\subsection{Impact of $N$}

When $N$ changes, the number of summation terms in (\ref{eq:chp3_A_0L_2}) as well as its own term envelope $|A_{0,n}|$, are also influenced.
Therefore, we analyze $A_0$ with respect to $N$ for a given $\theta_B$ by obtaining another approximation of $A_0$.
Let $A_0$ in (\ref{eq:chp3_A_0L_2}) have $A_{0,n}=4\pi\frac{N-n}{N}q_n$, where
$\{q_n\}$ is an series for given $\theta_B$ and $n\in\mathbb{N}^+$, i.e., 
\begin{align}
	q_n = J_0(k\mathit{\Delta d}n)\cos(k\mathit{\Delta d}n\sin\theta_B).
\end{align}

Examples of $\{q_n\}$ when $\mathit{\Delta d}=0.5\lambda$ are illustrated in Fig.\,\ref{fig:chp3_besselj_ULA_2}.
For the three different values of $\theta_B\in\{0^{\circ},30^{\circ},60^{\circ}\}$, it can be seen in Fig.\,\ref{fig:chp3_besselj_ULA_2} that the behavior of $\{q_n\}$ differs greatly.
When $\theta_B=0^{\circ}$, $q_n=J_0(n\pi)$ are discrete samples of $J_0(x)$.
When $\theta_B=30^{\circ}$, $q_n=J_0(n\pi)\cos(\frac{n\pi}{2})$ is zero for odd $n$; and $(-1)^{n/2}J_0(x)$ for even $n$.
When $\theta_B=60^{\circ}$, $q_n=J_0(n\pi)\cos(\frac{n\sqrt{3}\pi}{2})$.

\begin{figure}
\centering
\includegraphics[width=3.4in]{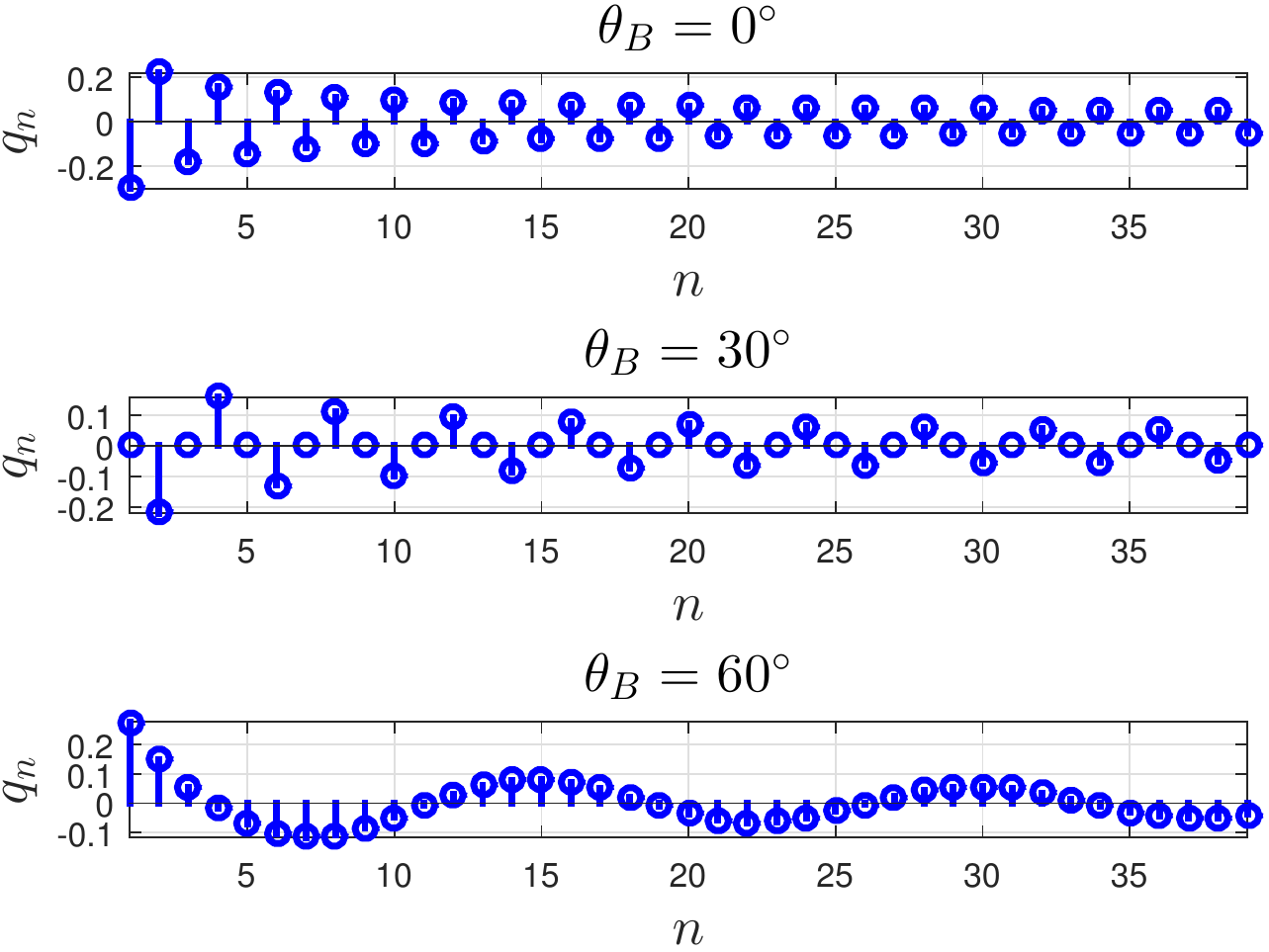}
\caption{Behavior of $q_n$ versus for $\theta_B\in\{0^{\circ},30^{\circ},60^{\circ}\}$.}
\label{fig:chp3_besselj_ULA_2}
\end{figure}

When $N$ is sufficiently large, $\frac{N-n}{N}$ becomes negligible for larger $n$; $q_n$ also approaches zero as $n$ increases, as illustrated in Fig.\,\ref{fig:chp3_besselj_ULA_2}.
In this case, only $A_{0,n}$, for $n\leq N_{up}\leq N-1$ needs to be considered.
Thus, the asymptotic expression when $N\to\infty$ can be expressed by
\begin{align}\label{eq:asymp_N}
	\lim_{N\to\infty}A_0&\approx 2\pi+4\pi\sum_{n=1}^{N_{up}}\frac{N-n}{N}q_n \nonumber \\
	   &= 2\pi+4\pi\sum_{n=1}^{N_{up}}q_n. 
\end{align}
The particular value of $N_{up}$, larger than which $q_n$ is negligible, is subject to practical requirement.
According to (\ref{eq:asymp_N}), we can asymptotically have
\begin{align}
	\lim_{N\to\infty} \bar{p}_{up} \approx 1-\text{exp}\{-\frac{\lambda_ec_0}{2}(2\pi+4\pi\sum_{n=1}^{N_{up}}q_n)\},
\end{align}
where $|q_n|\ll 1$ for $n>N_{up}$.

Fig.\,\ref{fig:chp3_p_N_De_ULA} depicts the impact of $N$ on $\bar{p}_{up}$ for various $\theta_B\in\{0^{\circ},30^{\circ},60^{\circ}\}$.
It can be seen from this figure that  when $N$ increases, $\bar{p}_{up}$ fluctuates at different rate for different $\theta_B$.
In addition, it can be observed that for any $\theta_B$, $\bar{p}_{up}$ approaches to a fixed value when $N$ grows sufficiently large.
This validates the asymptotic expression in (\ref{eq:asymp_N}).

\begin{figure}
\centering
\includegraphics[width=3.4in]{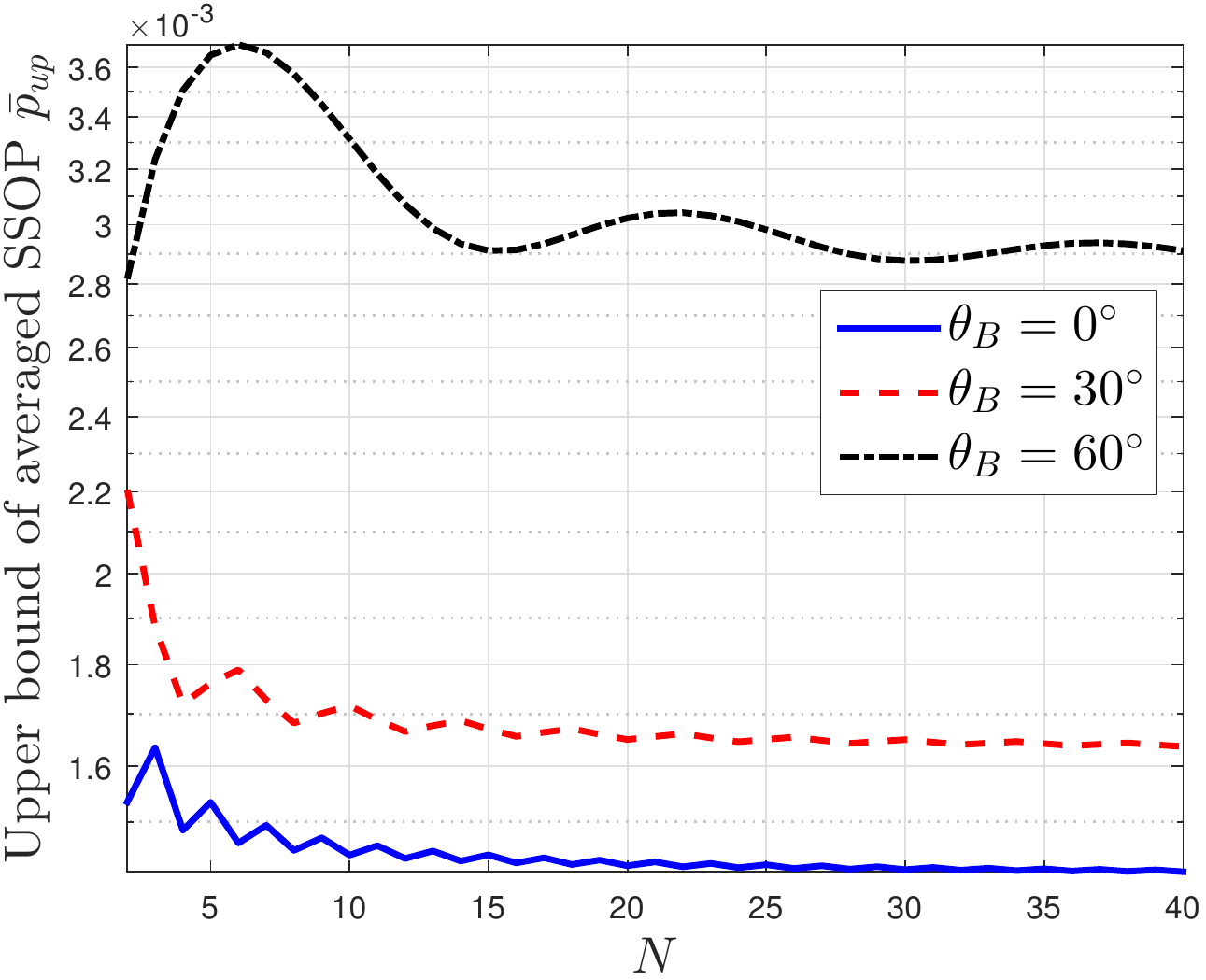}
\caption{$\bar{p}_{up}$ versus $N$ for Bob's angle $\theta_B\in\{0^{\circ},30^{\circ},60^{\circ}\}$. $\mathit{\Delta d}=0.5\lambda$, $P_t/\sigma_n^2=15$\,dB, $R_B=3.4594$\,bps/Hz, $R_s=1$\,bps/Hz, $\lambda_e=1\times10^{-4}$}
\label{fig:chp3_p_N_De_ULA}
\end{figure}

\section{Simulations and Numerical Results}\label{sec:analysis2}

In this section, we provide simulations and numerical results for $\bar{p}$ and $\bar{p}_{up}$ of the ER based beamforming over the Rician channel with any $K\ge 0$ and $\beta=\{2,3,4,5,6\}$ with respect to $N$ and $\theta_B$.

\subsection{SSOP and Its Upper Bound}
\label{chp3:result:wier}

In (\ref{eq:chp3_meanSSOP_up_2}), $\bar{p}_{up}$ is positively correlated with $\Big[\frac{c_0K}{2\pi(K+1)}A_0+\frac{c_0}{K+1}\Big]^{\frac{2}{\beta}}$.
For any fixed $\beta$ and $K$, $\bar{p}_{up}$ also has a positive relationship with $A_0$. 
Thus, the conclusions that are reached about $A_0$ regarding to the impact of $N$ and $\theta_B$ also apply to $\bar{p}_{up}$ for different $\beta$ and $K$.

For convenience, let $A_1$ denote $\frac{c_0K}{2\pi(K+1)}A_0+\frac{c_0}{K+1}$.
When $\beta$ increases from $2$ to $6$, $A_1^{\frac{2}{\beta}}$ decreases, because $A_1$ is generally larger than $1$.
It is also noticed that when $A_0=2\pi$, $K$ factor disappears, i.e., $A_1=c_0$.
When $A_0<2\pi$, the larger $K$ is, the smaller $A_1$ (i.e., $\bar{p}_{up}$) is; 
when $A_0>2\pi$, the larger $K$ is, the larger $A_1$ (i.e., $\bar{p}_{up}$) is.

In Fig.\,\ref{fig:chp3_p_up_K_beta}, the examples of $\bar{p}_{up}$ for different $K$ and $\beta$ are given for three typical values of $A_0$, i.e., $4.1326$, $2\pi$ and $15.3761$, which corresponds to $\theta_B=0^{\circ}$, $48.35^{\circ}$ and $90^{\circ}$ when $N=8$.
The logarithm scale is used to clearly show the ranges of $\bar{p}_{up}$ and $K$.
It can be seen that, when $\beta$ increases, $\bar{p}_{up}$ drops.
 For fixed $\beta$, $\bar{p}_{up}$ increases, stays unchanged or decreases depending on the value of $A_0$.

\begin{figure}
\centering
\includegraphics[width=3.4in]{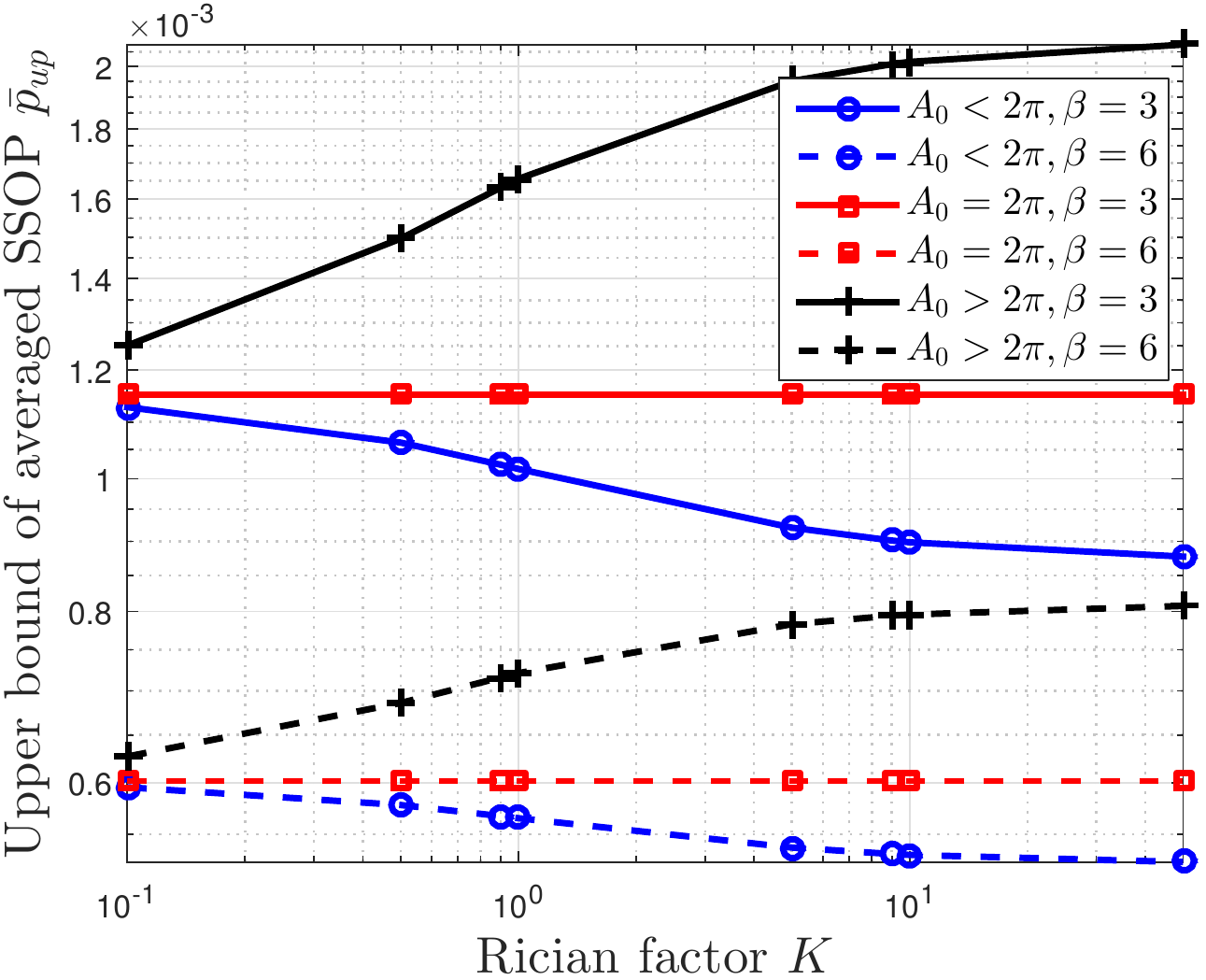}
\caption{$\bar{p}_{up}$ for different values of $A_0$, $K$ and $\beta$. $P_t/\sigma_n^2=40$\,dB, $R_B=3.4594$\,bps/Hz, $R_s=1$\,bps/Hz, $\lambda_e=1\times10^{-4}$}
\label{fig:chp3_p_up_K_beta}
\end{figure}

The range of $K$ in linear scale is from $0.01$ to $50$.
When $K=0.01$, the Rician channel approaches the Rayleigh channel ($K=0$).
When $K=50$, the Rician channel approaches the deterministic channel ($K=\infty$).
It can be seen that for fixed $\beta$, $\bar{p}_{up}$ is a constant for $K=0$ and is irrelevant to $A_0$ (nor $N$, $\theta_B$), as shown in (\ref{eq:chp3_meanSSOP_Ra}) and (\ref{eq:chp3_SSOP_Ra_up}).
When $K>10$, $\bar{p}_{up}$ approaches to a certain value that depends on $A_0$ which in turn depends on $N$ and $\theta_B$.

The above analysis of the properties of $\bar{p}_{up}$ serves as a coarse guidance for that of $p$. 
In the following, precise numerical results are used to show the properties of $\bar{p}$, which cannot be easily analyzed according to (\ref{eq:chp3_meanSSOP_Ri_2}).
First, the simulation results are provided to validate the expressions of $\bar{p}$ in (\ref{eq:chp3_meanSSOP_Ri_2}) to (\ref{eq:chp3_meanSSOP_Ra}) which are derived from the expression in (\ref{eq:chp3_meanSSOP_Ri_0}) which contains Gaussian random variables via $|\tilde{h}^2|$ according to (\ref{eq:chp3_A1}) and (\ref{eq:chp3_h_tilde_square}).

We choose $K=10$ and $\beta=3$ as an example to compare the numerical results based on the expression in (\ref{eq:chp3_meanSSOP_Ri_2}) and the simulation results based on the expression in (\ref{eq:chp3_meanSSOP_Ri_0}).
For the simulations, $1\times10^4$ samples are generated for $g_{Re}$ and $g_{Im}$ in (\ref{eq:chp3_h_tilde_square}).
The simulation and numerical results plotted in Fig.\,\ref{fig:chp3_p_DoE_numerical_and_simulation} show a good match between them, which verifies the validity of the expressions in (\ref{eq:chp3_meanSSOP_Ri_2}) to (\ref{eq:chp3_meanSSOP_Ra}).

\begin{figure}
\centering
\includegraphics[width=3.4in]{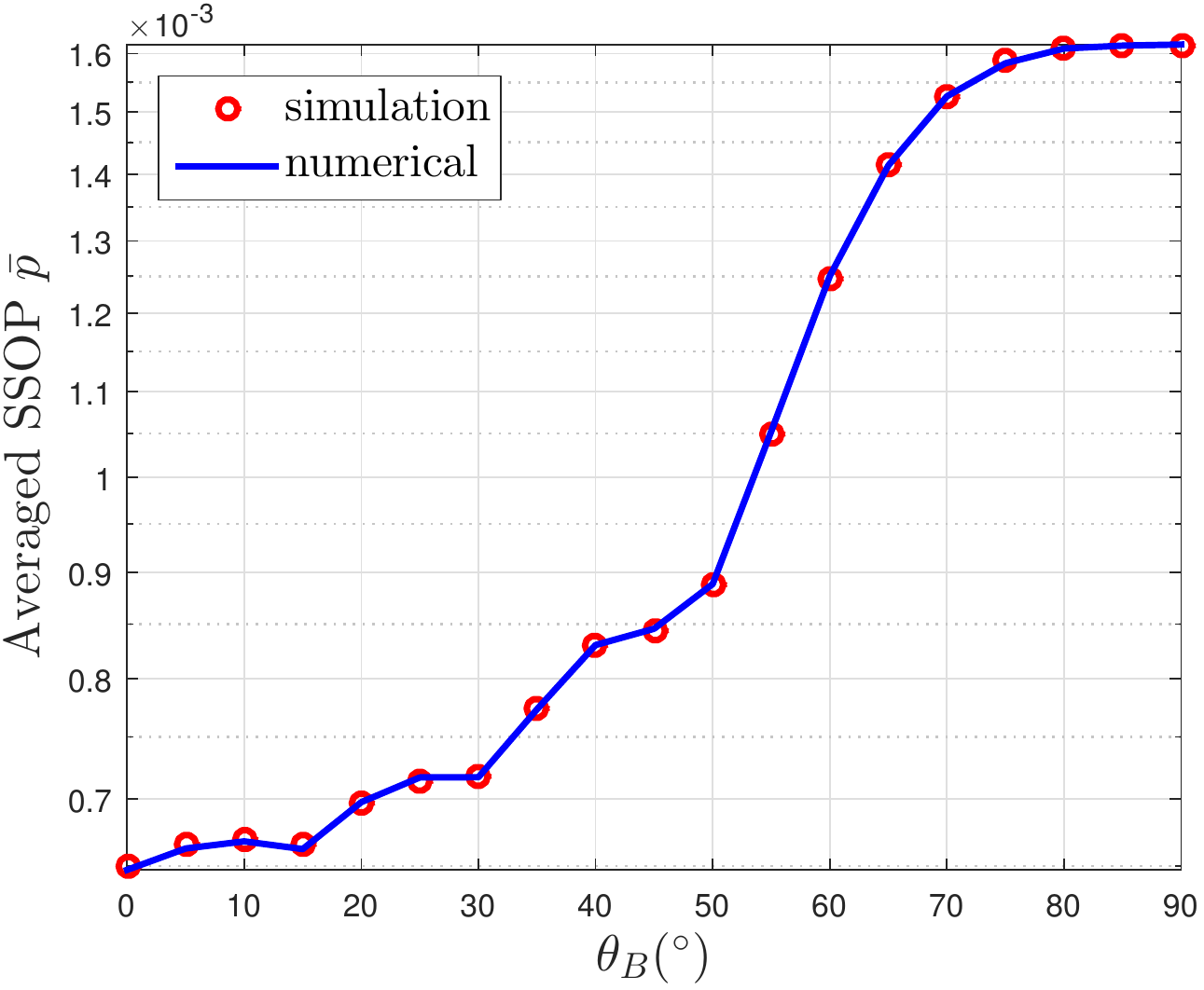}
\caption{Simulation and numerical results for $\bar{p}$ versus $\theta_B$; $K=10$, $\beta=3$, $P_t/\sigma_n^2=15$\,dB, $R_B=3.4594$\,bps/Hz, $R_s=1$\,bps/Hz, $\lambda_e=1\times10^{-4}$}
\label{fig:chp3_p_DoE_numerical_and_simulation}
\end{figure}

An example of $\bar{p}$ versus $\theta_B$ for $\beta=3$ and $N=8$ is given in Fig.\,\ref{fig:chp3_p_and_bounds_DoE_beta_3_L}.
$\beta=3$ is a typical value for some indoor scenarios such as home and factory\,\cite{goldsmith2005wireless}.
Typical values of $K$ are chosen as 0, 1, 10 and $\infty$.
In addition, $\bar{p}_{up}$ is also shown.

\begin{figure}
\centering
\includegraphics[width=3.4in]{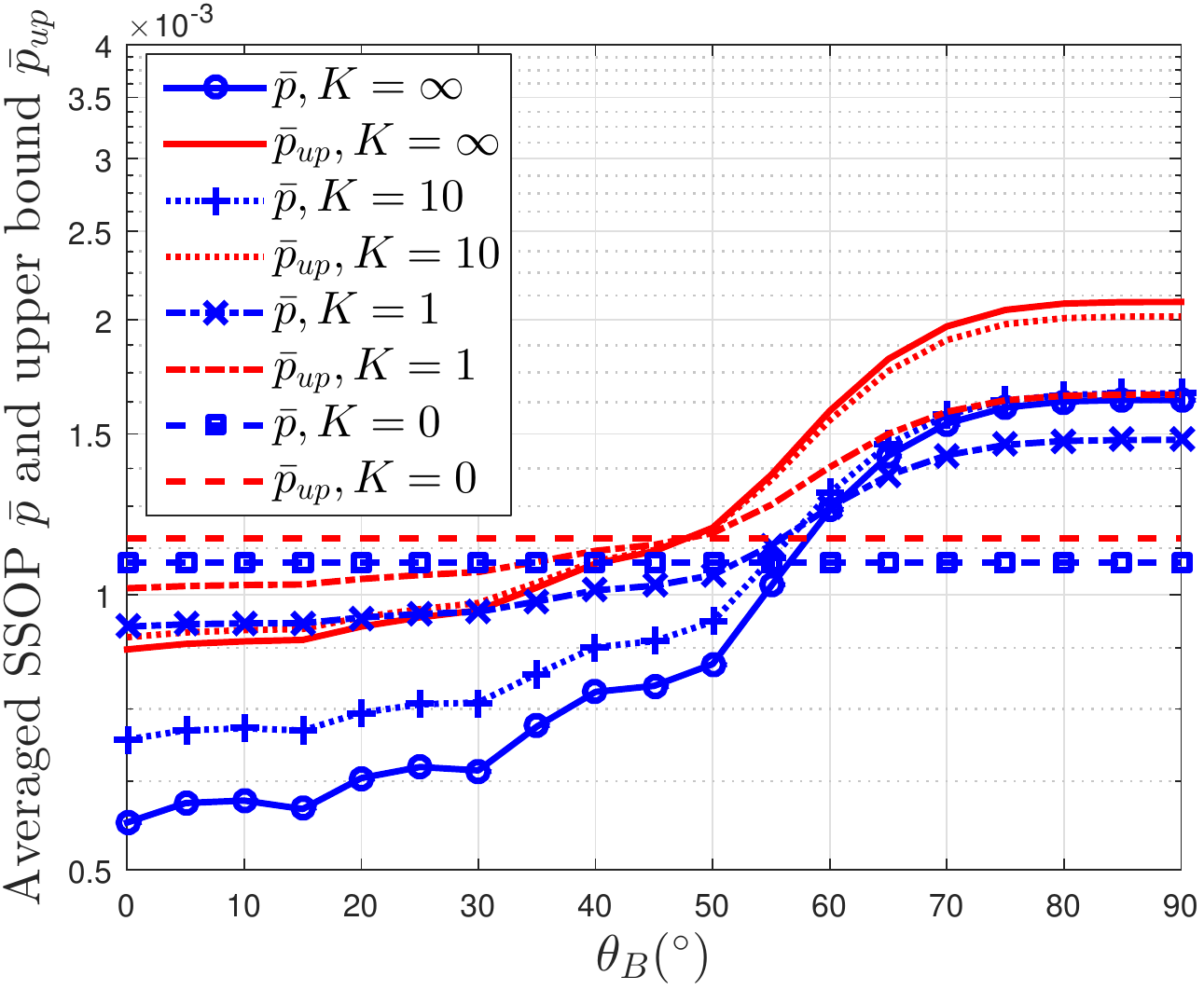}
\caption{$\bar{p}$ and $\bar{p}_{up}$ versus Bob's angle $\theta_B$ for different $K$ when $\beta=3$, $N=8$, $\mathit{\Delta d}=0.5\lambda$. $P_t/\sigma_n^2=15$\,dB, $R_B=3.4594$\,bps/Hz, $R_s=1$\,bps/Hz, $\lambda_e=1\times10^{-4}$}
\label{fig:chp3_p_and_bounds_DoE_beta_3_L}
\end{figure}

It can be seen that $\bar{p}$ and $\bar{p}_{up}$ increase in the range $\theta_B\in[0,\frac{\pi}{2}]$, except for $K=0$.
When $K=0$, the curves are flat because $\bar{p}$ and $\bar{p}_{up}$ are irrelevant to $\theta_B$, according to (\ref{eq:chp3_meanSSOP_Ra}) and (\ref{eq:chp3_SSOP_Ra_up}).
By comparing $\bar{p}_{up}$ and $\bar{p}$, it can be observed that the upper bound reflects the trend very well.
It can also be seen that for both $\bar{p}$ and $\bar{p}_{up}$, the curve for $K=10$ is closer to that for $K=\infty$, while the curve for $K=1$ is closer to that for $K=0$.


For completeness, Fig.\,\ref{fig:chp3_p_and_bounds_N_beta_3_L} shows an example of $\bar{p}$ and $\bar{p}_{up}$ versus $N$ for $\beta=3$ and $\theta_B=0^{\circ}$.
It can be seen that $\bar{p}$ and $\bar{p}_{up}$ decrease to different floor levels depending on $K$.
The same behavior has been shown in Fig.\,\ref{fig:chp3_p_N_De_ULA} where $K=\infty$ and $\beta=2$.
However, it can also be seen that $\bar{p}$ converges with a much slower speed, leading to an increasing larger gap between $\bar{p}$ and $\bar{p}_{up}$ as $N$ increases.

\begin{figure}
\centering
\includegraphics[width=3.4in]{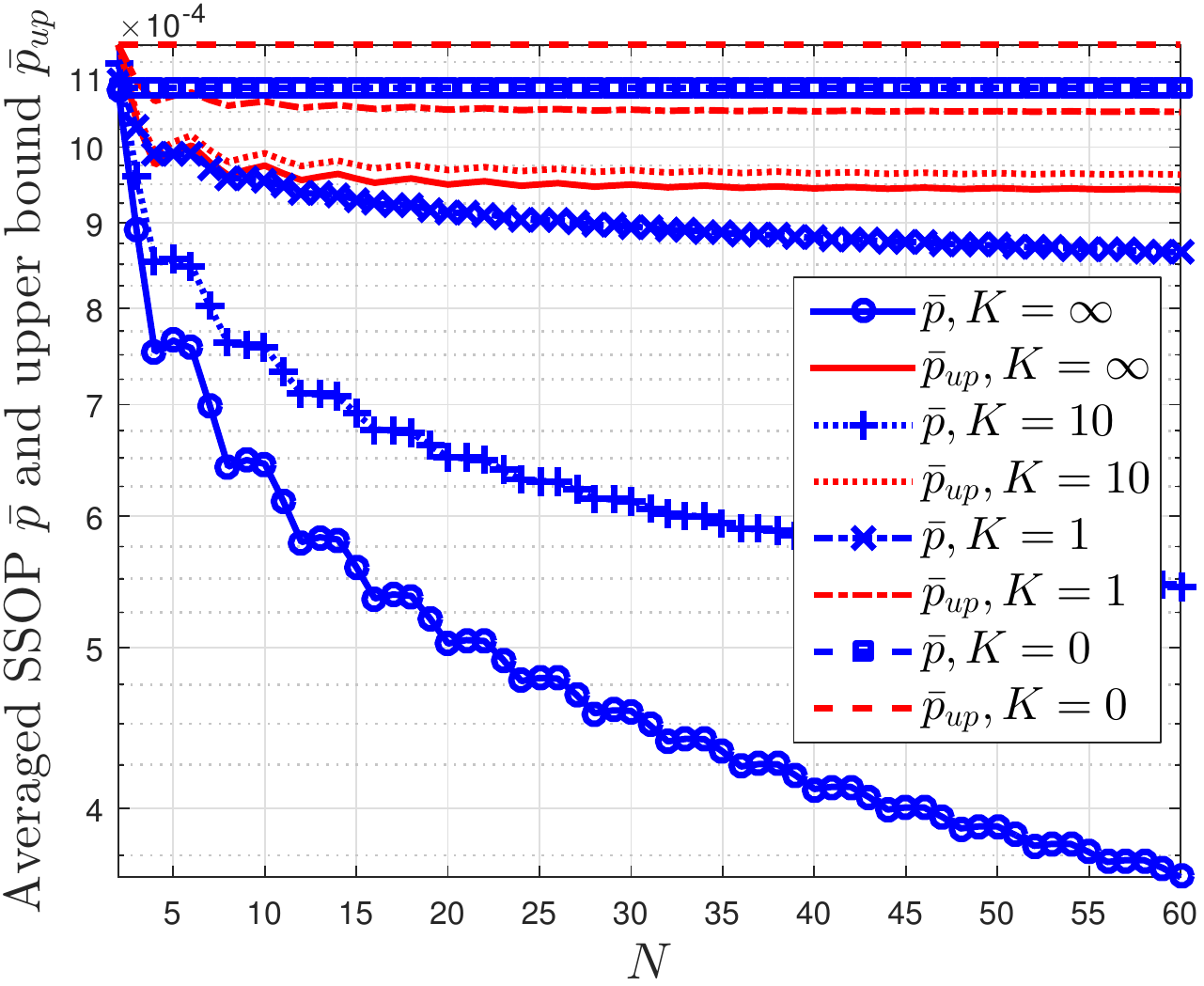}
\caption{$\bar{p}$ and $\bar{p}_{up}$ versus number of elements $N$ for different $K$ when $\beta=3$, $\theta_B=0^{\circ}$, $\mathit{\Delta d}=0.5\lambda$. $P_t/\sigma_n^2=15$\,dB, $R_B=3.4594$\,bps/Hz, $R_s=1$\,bps/Hz, $\lambda_e=1\times10^{-4}$}
\label{fig:chp3_p_and_bounds_N_beta_3_L}
\end{figure}

In summary, the properties of $A_0$ with respect to $N$ and $\theta_B$ can be extended to $\bar{p}_{up}$.
As for $\bar{p}$, while $\bar{p}$ has similar properties to $A_0$ with respect to $N$ and $\theta_B$, the gaps between $\bar{p}$ and $\bar{p}_{up}$ increase as $N$.
Therefore, in the next section, the tightness of $\bar{p}_{up}$ will be examined.

\subsection{Tightness of Upper Bound}
\label{chp3:result:mnbv}

In this section, the tightness of the upper bound is examined via numerical results with respect to $(K,\beta, N,\theta_B)$.
An example of $\bar{p}$ and $\bar{p}_{up}$ for different $K$ and $\beta$ with $N=8$ and $\theta_B=0^{\circ}$ is shown in Fig.\,\ref{fig:chp3_p_and_p_up_K_beta}.
At lower region of $K$, the channel approaches the Rayleigh channel.
Thus, $\bar{p}$ and $\bar{p}_{up}$ converge to the certain values that only depend on $\beta$ according to (\ref{eq:chp3_meanSSOP_Ra}) and (\ref{eq:chp3_SSOP_Ra_up}).
At higher region of $K$, the channel approaches the deterministic channel.
$\bar{p}$ and $\bar{p}_{up}$ converge to the certain values that depend on $\beta$ and $G(\theta,\theta_B)$,  according to (\ref{eq:chp3_SSOP_De}) and (\ref{eq:chp3_SSOP_De_up}).

It can also be seen that when $\beta=2$, the curves for $\bar{p}$ and $\bar{p}_{up}$ emerge as $K$ increases, which corresponds to $\bar{p}=\bar{p}_{up}$ for the deterministic channel.
For other values of $\beta$, as $K$ increases, the gaps between $\bar{p}$ and $\bar{p}_{up}$ increases.

\begin{figure}
\centering
\includegraphics[width=3.4in]{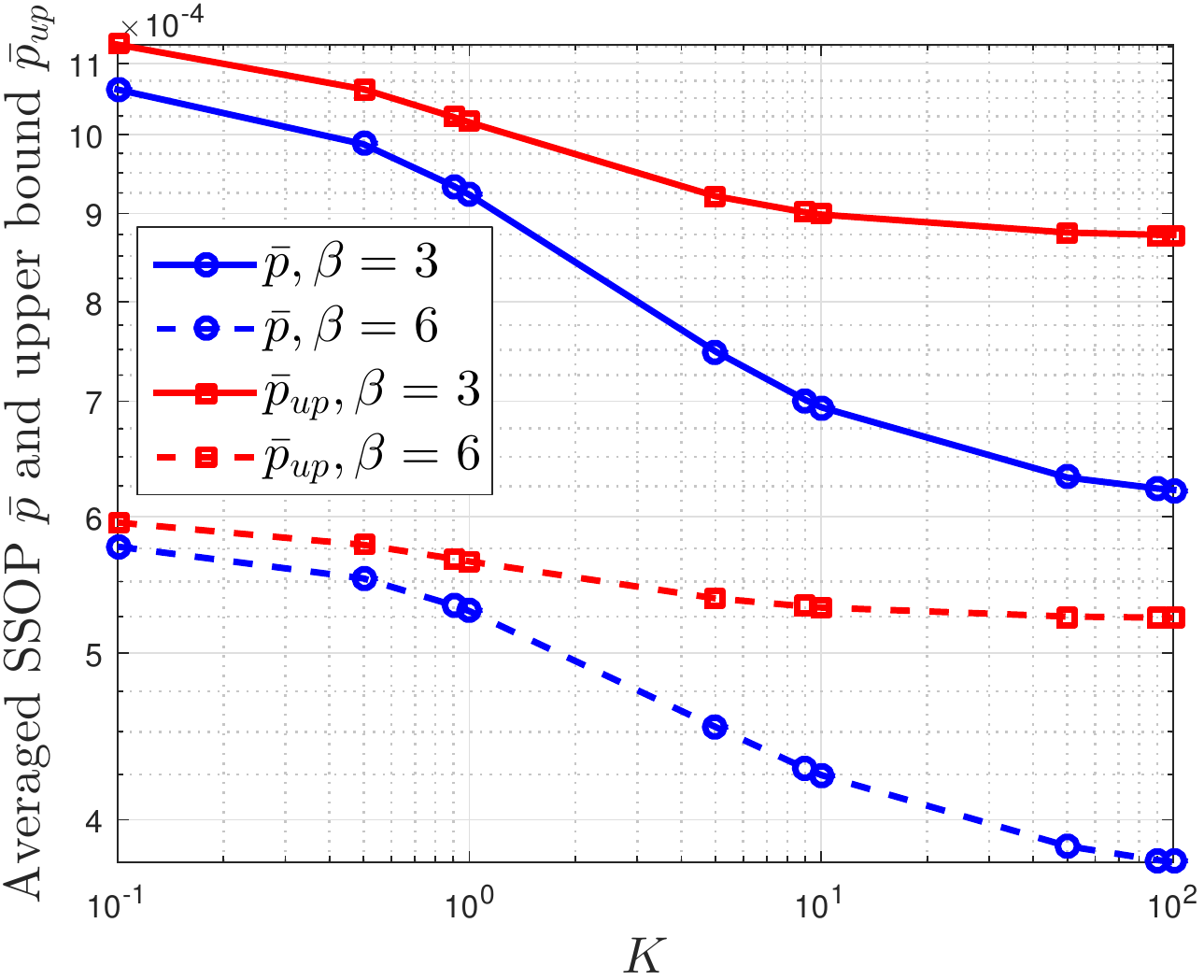}
\caption{$\bar{p}$ and $\bar{p}_{up}$ for different $K$ and $\beta$ when $N=8$, $\theta_B=0^{\circ}$, $\mathit{\Delta d}=0.5\lambda$. $P_t/\sigma_n^2=15$\,dB, $R_B=3.4594$\,bps/Hz, $R_s=1$\,bps/Hz, $\lambda_e=1\times10^{-4}$}
\label{fig:chp3_p_and_p_up_K_beta}
\end{figure}

In this section, the ratio between $\bar{p}_{up}$ and $\bar{p}$ is used to measure the tightness of $\bar{p}_{up}$. 
Let $\eta$ denote the ratio,
\begin{align}\label{eq:chp3_eta}
	\eta=\frac{\bar{p}_{up}}{\bar{p}}.
\end{align}
$\eta\geq 1$.
The smaller value of $\eta$, the tighter $\bar{p}_{up}$ is.
In Fig.\,\ref{fig:chp3_p_and_p_up_K_beta}, it can be deduced that $\eta$ will take the minimum value at $K=0$ and approach the maximum value at $K=\infty$.
Thus, in the following, the extreme cases $K=0$ and $K=\infty$ are used to study the range of $\eta$ for different $N$, $\theta_B$ and $\beta$.

In Fig.\,\ref{fig:chp3_eta_DoE_bounds_ULA}, $\eta$ is plotted against $\theta_B$ for $K=0$ and $K=\infty$ for all $\beta$.
The ULA has $N=8$ elements and $\mathit{\Delta d}=0.5\lambda$.
For Rayleigh channel, both $\bar{p}$ and $\bar{p}_{up}$ are irrelevant to $\theta_B$, thus $\eta$ is flat across $\theta_B\in[0,90^{\circ}]$. 
For the deterministic channel when $\beta=2$, $\eta=1$; 
when $\beta>2$, $\eta$ in general decrease with $\theta_B$.

\begin{figure}
\centering
\includegraphics[width=3.4in]{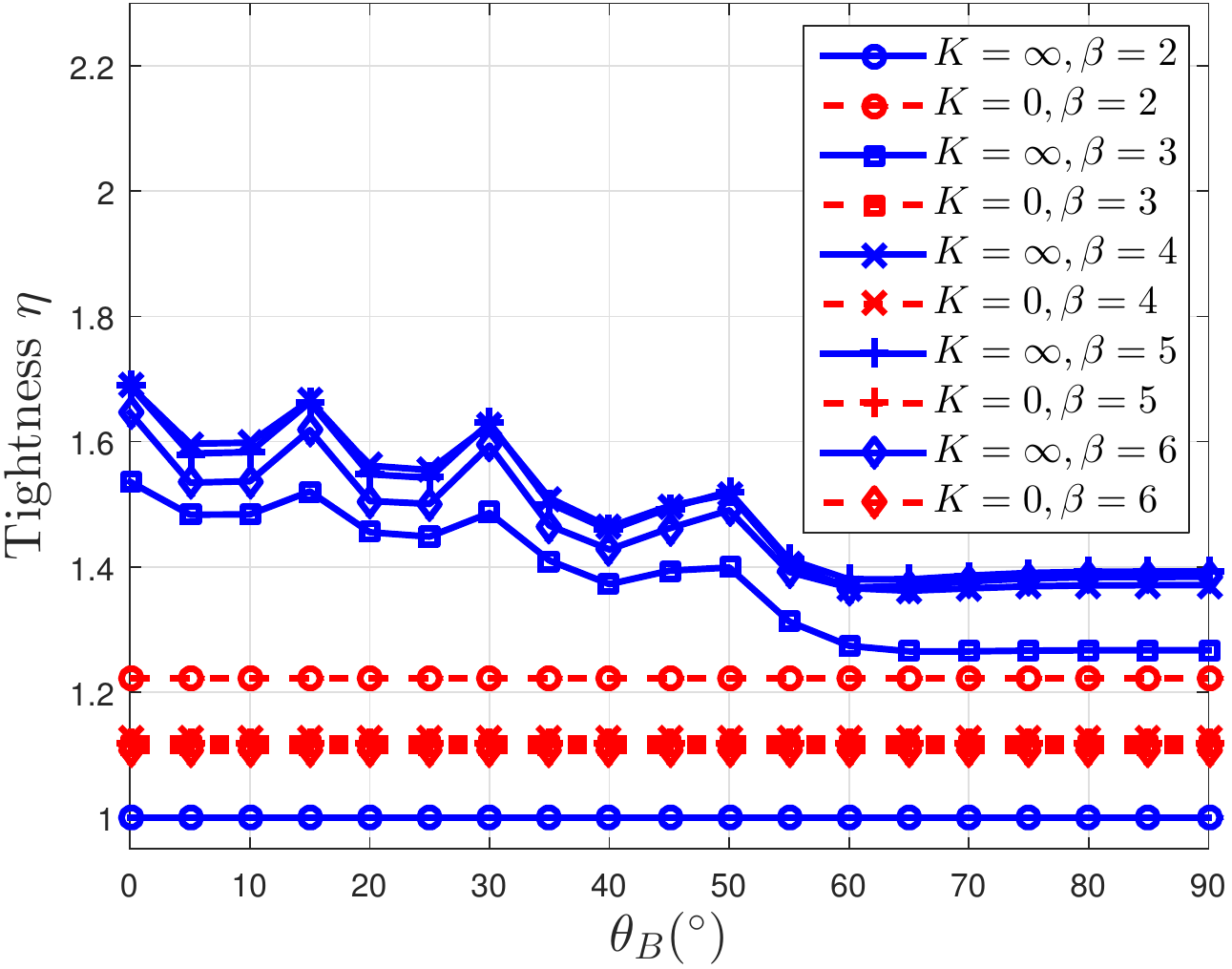}
\caption{$\eta$ versus Bob's angle $\theta_B$ for deterministic and Rayleigh channels for all $\beta$, number of elements is $N=8$}
\label{fig:chp3_eta_DoE_bounds_ULA}
\end{figure}

Comparing the curves for both the deterministic and the Rayleigh channels, it is noticed that when $\beta>2$, the ratios are located closely in a cluster.
However, there does not exist monotonic relationship between $\eta$ and $\beta$.
For example, when $\beta=6$, $\eta$ for the deterministic channel is smaller than that when $\beta=4$.

In Fig.\,\ref{fig:chp3_eta_N_bounds_ULA}, $\eta$ is plotted against $N$ for $K=\{0,\infty\}$ and $\beta\in\{2,3,4,5,6\}$.
The ULA has $\mathit{\Delta d}=0.5\lambda$ and $\theta_B=0^{\circ}$.
For the Rayleigh channel, $\eta$ is flat across $N$ for all $\beta$.
For the deterministic channel, $\eta$ in general increases with $N$ when $\beta>2$, which verifies the observation from Fig.\,\ref{fig:chp3_p_and_bounds_N_beta_3_L}.

\begin{figure}
\centering
\includegraphics[width=3.4in]{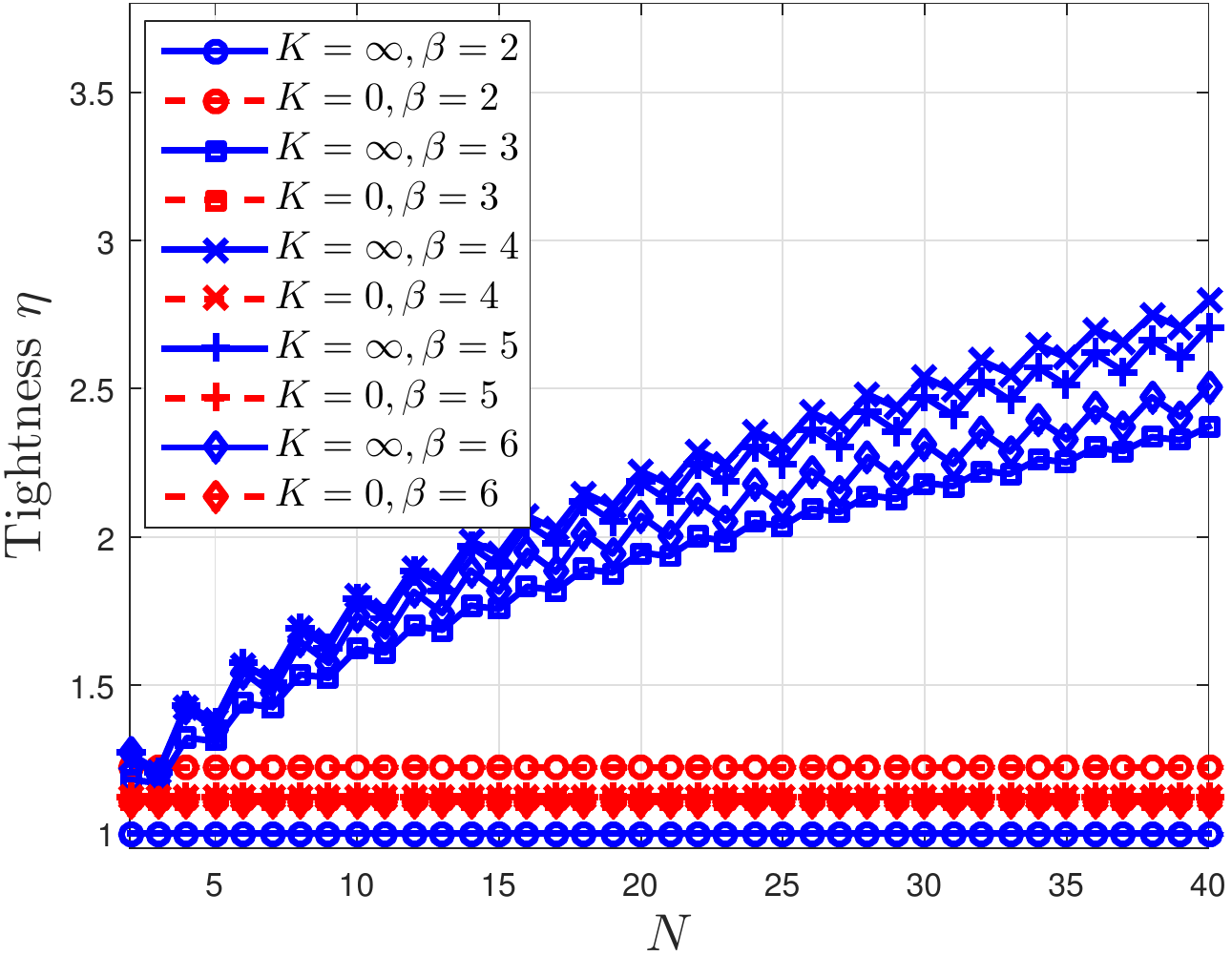}
\caption{$\eta$ versus $N$ for deterministic and Rayleigh channels for all $\beta$, $\theta_B=0^{\circ}$}
\label{fig:chp3_eta_N_bounds_ULA}
\end{figure}

In summary, when $\beta=2$, $\eta$ decreases with $K$ till the minimum value $\eta=1$;
when $\beta>2$, $\eta$ increases with $K$ till certain value that depends on $N$ and $\theta_B$, and the values of $\eta$ for different $\beta$ stay in a cluster.
For given $\beta$ and $K$, $\eta$ generally decreases with $\theta_B$ and increases with $N$.
In a lower region of $N$, e.g., $N<10$, the value of $\eta$ is smaller than 2.


\section{Conclusions}\label{sec:concl}

This paper has investigated secure wireless communications whereby a ULA in Alice communicates to Bob in the presence of PPP distributed Eves. Particularly, we mathematically defined ER to characterize spatial secrecy outage event and proposed the ER based beamforming over a Rician fading channel. As for the analysis of the ER, the analytic expression of the pattern area was also derived in form of Bessel function and two different approximations were adopted to analyze how the Bob’s angle and the number of element of the ULA quantify the ER. Using the ER, the SSOP was defined and the SSOP performance was evaluated, allowing the derivation of its exact and upper bound closed-form expressions. The impact of the array parameters on the SSOP was discussed to find that the SSOP increases dramatically with increasing Bob’s angle; decreases with reducing ER; and approaches certain level with increasing number of antenna elements. Simulations and the numerical results validated our analysis and examined the tightness of the upper bound expressions. Since the definitions of the ER and the SSOP were generalized to be applicable to any array type, the results can be useful to various antenna array types in future wireless security systems.

\section*{Acknowledgment}
The authors gratefully acknowledge support from the US-Ireland R\&D Partnership USI033 `‘WiPhyLoc8’' grant involving Rice University (USA), University College Dublin (Ireland) and Queen’s University Belfast (N. Ireland).

\ifCLASSOPTIONcaptionsoff
  \newpage
\fi
\bibliographystyle{IEEEtran}
\bibliography{IEEEabrv,bibliography}

\appendices

\section{Proof of Theorem\,\ref{th:sec4_p_up}}
\label{appdx:bessel:owejg}

According to (\ref{eq:chp3_meanSSOP_Ri_0}) and (\ref{eq:chp3_JI_1}), it can be derived that
\begin{align}\label{eq:chp3_meanSSOP_up_inequality}
	 \bar{p}=1-\mathbb{E}_{|\tilde{h}|}[e^{-\lambda_eA}]\leq 1-e^{-\lambda_e\mathbb{E}_{|\tilde{h}|}[A]}.
\end{align}
Notice that $A$ depends on random variable $\tilde{h}$ and is not constant, except for $K=\infty$.
Thus, the equality holds only for deterministic channels.

To solve (\ref{eq:chp3_meanSSOP_up_inequality}), assume that $\theta\sim \mathcal{U}(0,2\pi)$.
According to (\ref{eq:chp3_A3}), $A$ in (\ref{eq:chp3_meanSSOP_up_inequality}) can be converted into
\begin{align}\label{eq:chp3_meanSSOP_up_inequality2}
A=2\pi\frac{1}{2}\int_0^{2\pi}\frac{1}{2\pi}X_{\theta}^{\frac{2}{\beta}}\,\mathrm{d}\theta
 =\pi\mathbb{E}_{\theta}[X_{\theta}^{\frac{2}{\beta}}].
\end{align}
According to (\ref{eq:chp3_JI_2}), (\ref{eq:chp3_meanSSOP_up_inequality2}) is bounded by
\begin{align}\label{eq:chp3_JI_theta}
A\leq \pi(\mathbb{E}_{\theta}[X_{\theta}])^{\frac{2}{\beta}}
 =\pi \Big(\int_0^{2\pi}\frac{1}{2\pi}X_{\theta}\,\mathrm{d}\theta \Big)^{\frac{2}{\beta}}.
\end{align}
In the inequality, the equality holds when $\beta=2$ for any $K$.

According to (\ref{eq:chp3_meanSSOP_up_inequality}) and (\ref{eq:chp3_JI_theta}), it can be derived that
\begin{align}\label{eq:chp3_meanSSOP_up_inequality3}
	\mathbb{E}_{|\tilde{h}|}[A] 
	\leq \pi\mathbb{E}_{|\tilde{h}|}\Big[\Big(\int_0^{2\pi}\frac{1}{2\pi}X_{\theta}\,\mathrm{d}\theta\Big)^{\frac{2}{\beta}}\Big].
\end{align}
Then applying (\ref{eq:chp3_JI_2}) and (\ref{eq:chp3_meanSSOP_up_inequality3}), it can be derived that
\begin{align}\label{eq:chp3_inequality_jvje}
	\pi\mathbb{E}_{|\tilde{h}|}\Big[\Big(\int_0^{2\pi}\frac{1}{2\pi}X_{\theta}\,\mathrm{d}\theta\Big)^{\frac{2}{\beta}}\Big]
	\leq \pi\Big(\mathbb{E}_{|\tilde{h}|}\Big[\int_0^{2\pi}\frac{1}{2\pi}X_{\theta}\,\mathrm{d}\theta\Big]\Big)^{\frac{2}{\beta}}. 
\end{align}
Exchanging the integral and $\mathbb{E}_{|\tilde{h}|}$, then substituting $X_{\theta}=c_0|\tilde{h}|^2$, it can be derived that
\begin{align}\label{eq:chp3_meanA_up}
	\mathbb{E}_{|\tilde{h}|}[A] \leq \pi\Big(\frac{c_0}{2\pi}\int_0^{2\pi}\mathbb{E}_{|\tilde{h}|}[|\tilde{h}|^2]\,\mathrm{d}\theta\Big)^{\frac{2}{\beta}}.
\end{align}
Notice that when $\beta=2$, the equality holds.

Apply (\ref{eq:chp3_meanA_up}) to (\ref{eq:chp3_meanSSOP_up_inequality}) then obtain
\begin{align}
	\bar{p}&	\leq 1-e^{-\lambda_e\mathbb{E}_{|\tilde{h}|}[A]} \nonumber \\
	&\leq 1-\text{exp}\Big[-\lambda_e\pi\Big(\frac{c_0}{2\pi}\int_0^{2\pi}\mathbb{E}_{|\tilde{h}|}[|\tilde{h}|^2]\,\mathrm{d}\theta\Big)^{\frac{2}{\beta}}\Big].
\end{align}
The upper bound $\bar{p}_{up}$ can be expressed by
\begin{align}\label{eq:chp3_meanSSOP_up_1}
	\bar{p}_{up}=1-\text{exp}\Big[-\lambda_e\pi\Big(\frac{c_0}{2\pi}\int_0^{2\pi}\mathbb{E}_{|\tilde{h}|}[|\tilde{h}|^2]\,\mathrm{d}\theta\Big)^{\frac{2}{\beta}}\Big].
\end{align}
According to (\ref{eq:chp3_h_tilde_square}), $\mathbb{E}_{|\tilde{h}|}[|\tilde{h}|^2]=\frac{KG^2(\theta,\theta_B)+1}{K+1}$.
Substituting the previous result into (\ref{eq:chp3_meanSSOP_up_1}), (\ref{eq:chp3_meanSSOP_up_2}) can be obtained.

For special cases, take the limit of $K\to\infty$ and $K\to 0$, (\ref{eq:chp3_SSOP_De}) and (\ref{eq:chp3_meanSSOP_Ra}) can be obtained.


For the ULA, (\ref{eq:chp3_A_0}) can be further derived according to (\ref{eq:chp3_AF_ULA}).
\begin{align}
	A_0&=\int_0^{2\pi} \frac{1}{N}\sum_{i,j} e^{jk\Delta d(\sin\theta_B-\sin\theta)(i-j)} \,\mathrm{d}\theta \nonumber \\
		 &=\frac{1}{N}\sum_{i,j} e^{jk\Delta d\sin\theta_B(i-j)}
	  \int_0^{2\pi} e^{-jk\Delta d\sin\theta(i-j)} \mathrm{d}\theta. \label{eq:appdx_bessel_aieow}
\end{align}
According to the integral representation of the Bessel function of the first kind, $J_n(x)=\frac{1}{2\pi}\int_{-\pi}^{\pi}e^{j(n\tau-x\sin\tau)}\mathrm{d}\tau$, 
(\ref{eq:appdx_bessel_aieow}) can be further derived by
\begin{align}\label{eq:appdx_bessel_euirue}
  A_0=\frac{2\pi}{N}\sum_{i,j} J_0(k\Delta d(i-j))e^{jk\Delta d(i-j)\sin\theta_B},
\end{align}
where $A_0$ is the summation of $N\times N$ terms.
To further simplify (\ref{eq:appdx_bessel_euirue}), each of which is denoted by $A_{0,i,j}$,
\begin{align}
	A_{0,i,j}=\frac{2\pi}{N} J_0(k\Delta d(i-j))e^{jk\Delta d(i-j)\sin\theta_B}.
\end{align}
Notice that the only variable across all $A_{0,i,j}$ is the difference $i-j$.
So let $n=i-j$ and it can be derived that
\begin{align}
	A_{0,n}=\frac{2\pi}{N} J_0(k\Delta dn)e^{jk\Delta dn\sin\theta_B}.
\end{align}
Then, all the values of $n$ that are associated with $A_{0,n}$ are mapped into a table shown in Fig.\,\ref{fig:appdx_bessel_table_ULA}.

\begin{figure}
\centering
\includegraphics[scale=1]{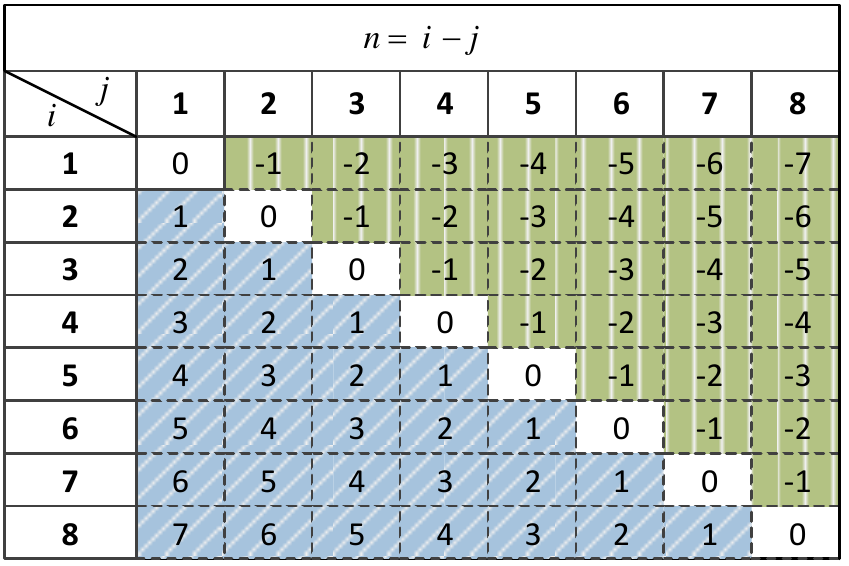}
\caption{Table for $A_{0,i,j}$ shows the symmetry regarding to the diagonal line $i=j$}
\label{fig:appdx_bessel_table_ULA}
\end{figure}

Observing the table in Fig.\,\ref{fig:appdx_bessel_table_ULA}, it is noticed that i) the terms of $A_{0,n}$ on the diagonal lines can be combined, because they are the same; 
ii) becuase $J_m(-x)=(-1)^mJ_m(x)$, the terms of $A_{0,n}$ that have the same absolute value of $n$ can be added
\begin{align}
	&A_{0,n}+A_{0,-n}\nonumber \\
	=&\frac{2\pi}{N} [J_0(k\Delta dn)e^{jk\Delta dn\sin\theta_B}+J_0(-k\Delta dn)e^{-jk\Delta dn\sin\theta_B}] \nonumber \\
	=&\frac{4\pi}{N}J_0(k\Delta dn)\cos(k\Delta dn\sin\theta_B).
\end{align}
In addition, when $n=0$, $J_0(0)=1$ and $e^{j0}=1$.
Thus, $A_{0,0}=\frac{2\pi}{N}$.
Now, sum up the terms of $A_{0,n}$ on each diagonal lines from $n=0$ to $p=N-1$ and obtain
\begin{align}\label{eq:appdx_bessel_ieur}
	A_0=2\pi+4\pi\sum_{n=1}^{N-1} \frac{N-n}{N}J_0(k\Delta dn)\cos(k\Delta dn\sin\theta_B).
\end{align}

\end{document}